\documentclass{article}

\usepackage[affil-it]{authblk}
\usepackage[usenames,dvipsnames]{xcolor}
\usepackage{amsfonts}
\usepackage{amsmath,amsthm,amssymb,dsfont}
\usepackage{enumerate}
\usepackage[english]{babel}
\usepackage{graphicx}	
\usepackage[caption=false]{subfig}
\usepackage[margin=3cm]{geometry}
\usepackage{url}
\usepackage{todonotes}
\usepackage{bbm}

\usepackage{tikz}
\usetikzlibrary{chains}
\usetikzlibrary{fit}
\usepackage{pgflibraryarrows}		
\usepackage{pgflibrarysnakes}		

\usepackage{epsfig}
\usetikzlibrary{shapes.symbols,patterns} 
\usepackage{pgfplots}

\usepackage{hyperref}
\hypersetup{colorlinks=true,citecolor=blue,linkcolor=blue,filecolor=blue,urlcolor=blue,breaklinks=true}

\usepackage{nicefrac}

\theoremstyle{plain}
\newtheorem{theorem}{Theorem}[section]
\newtheorem{lemma}[theorem]{Lemma}

\newtheorem{corollary}[theorem]{Corollary}
\newtheorem{proposition}[theorem]{Proposition}

\theoremstyle{definition}

\newcommand*{\cH}{\mathcal{H}}

\newcommand*{\cI}{\mathcal{I}}

\newcommand*{\cN}{\mathcal{N}}

\newcommand*{\cP}{\mathcal{P}}

\newcommand*{\cR}{\mathcal{R}}

\newcommand*{\cX}{\mathcal{X}}

\newcommand*{\N}{\mathbb{N}}
\newcommand*{\R}{\mathbb{R}}
\newcommand*{\C}{\mathbb{C}}

\newcommand*{\id}{\mathrm{id}}
\newcommand*{\poly}{\mathrm{poly}}

\newcommand*{\spec}{\mathrm{spec}}

\newcommand*{\tr}{\mathrm{tr}}
\newcommand*{\ket}[1]{| #1 \rangle}
\newcommand*{\bra}[1]{\langle #1 |}

\newcommand*{\ci}{\mathrm{i}} 
\newcommand*{\di}{\mathrm{d}} 

\newcommand{\norm}[1]{\left\lVert#1\right\rVert}

\allowdisplaybreaks    

\begin{document}

\title{\LARGE Multivariate Trace Inequalities}

\author[1]{David Sutter}
\author[2]{Mario Berta}
\author[3]{Marco Tomamichel}

\affil[1]{Institute for Theoretical Physics, ETH Zurich, Switzerland}
\affil[2]{Institute for Quantum Information and Matter, Caltech, USA}
\affil[3]{School of Physics, The University of Sydney, Australia}
\date{}

\maketitle


\begin{abstract}
We prove several trace inequalities that extend the Golden-Thompson and the Araki-Lieb-Thirring inequality to arbitrarily many matrices. In particular, we strengthen Lieb's triple matrix inequality. As an example application of our four matrix extension of the Golden-Thompson~inequality, we prove remainder terms for the monotonicity of the quantum relative entropy and strong sub-additivity of the von Neumann entropy in terms of recoverability. We find the first explicit remainder terms that are tight in the commutative case. Our proofs rely on complex interpolation theory as well as asymptotic spectral pinching, providing a transparent approach to treat generic multivariate trace inequalities.
\end{abstract}


\section{Introduction}

Trace inequalities are mathematical relations between different multivariate trace functionals. Often these relations are straightforward equalities if the involved matrices commute\,---\,and can be difficult to prove for the non-commuting case. 

Arguably one of the most powerful trace inequality is the celebrated \emph{Golden-Thompson (GT) inequality}~\cite{golden65,thompson65}. It states that for any two Hermitian matrices $H_1$ and $H_2$ we have
\begin{align} \label{eq_GToriginal}
\tr \exp(H_1 + H_2)  \leq \tr \exp(H_1) \exp(H_2) \ .
\end{align}
We note that in case $H_1$ and $H_2$ commute~\eqref{eq_GToriginal} holds with equality.
Inequality~\eqref{eq_GToriginal} has been generalized in various directions (see, e.g.,~\cite{breitenecker72,ruskai72,araki73,simon_book79,kilmek91,kosaki92,HP93,Li14}).
For example, it has been shown that it remains valid by replacing the trace with any unitarily invariant norm~\cite{segal69,lenard71,thompson71} and an extension to three non-commuting matrices was suggested in~\cite{Lieb73}.
The GT inequality has found applications ranging from statistical physics~\cite{thompson65} and random matrix theory~\cite{AW02,Tropp11} to quantum information theory~\cite{LieRus73_1,LieRus73}.

Straightforward extensions of this inequality to three matrices are incorrect, as discussed in Appendix~\ref{app_beta0}. In this work, for any $n\in \N$, Hermitian matrices $\{H_k\}_{k=1}^n$ and any $p \geq 1$, we show that
\begin{align} \label{eq_mainResGT_intro}
\log \norm{ \exp \left( \sum_{k=1}^n H_k \right) }_p \leq 
\int_{-\infty}^\infty \di t \, \beta_0(t)\, \log \norm{ \prod_{k=1}^n  \exp\bigl( (1+\ci t) H_k \bigr) }_p \ ,
\end{align}
where $\norm{\cdot}_p$ denotes the Schatten $p$-norm and $\beta_0$ is a fixed probability distribution on $\mathbb{R}$, independent of the other parameters. An extensive discussion of this result is provided in Section~\ref{sec_main}. The precise statement is given in Corollary~\ref{cor_GT_steinHirschman}. Note that the expression $\exp((1 + \ci t) H_k)$ decomposes as
$\exp(H_k)\exp(\ci t H_k)$ where the latter is a unitary rotation. 
 Since the Schatten $p$-norm is unitarily invariant it follows that the integrand in~\eqref{eq_mainResGT_intro} 
is independent of $t$ for $n = 2$.
Inequality~\eqref{eq_mainResGT_intro} thus constitutes an $n$-matrix extension of the GT inequality and further simplifies to~\eqref{eq_GToriginal} for $n=2$ and $p = 2$. For $n=3$ and $p=2$ our result strengthens 
\emph{Lieb's triple matrix inequality}~\cite{Lieb73}, as shown in Lemma~\ref{lem_LiebRep}. 

The GT inequality can be seen as a limiting case of the more general \emph{Araki-Lieb-Thirring (ALT) inequality}~\cite{Lieb76,araki90}. The latter states that, for any positive semi-definite matrices $A_1$ and $A_2$, and~$q > 0$, 
\begin{align} \label{eq_ALT_intro}
\tr\left(A_1^{\frac{r}{2}} A_2^r A_1^{\frac{r}{2}}\right)^{\frac{q}{r}}  \leq \tr \left(A_1^{\frac{1}{2}} A_2 A_1^{\frac{1}{2}}\right)^{q}   \quad \textnormal{if} \quad r \in (0,1] \,.
\end{align}
The inequality holds in the opposite direction for $r \geq 1$ by an appropriate substitution.\footnote{This can be seen by considering the substitution $A_i^r \leftarrow A_i$ for $i \in \{1,2\}$, $\frac{q}{r}\leftarrow q$, and $\frac{1}{r}\leftarrow r$.}
The GT inequality for Schatten $p$-norms is implied by the Lie-Trotter product formula in the limit $r \to 0$.
The ALT inequality has also been extended in various directions (see, e.g.,~\cite{kosaki92,ando94,wang95}).

In this work, we provide an $n$-matrix extension of the ALT inequality. For any $n\in \N$, positive semi-definite matrices $\{A_k\}_{k=1}^n$ and any $p \geq 1$, we show that
\begin{align}\label{eq_ALT_mainres1}
&\log \norm{ \left| \prod_{k=1}^n A_k^r \right|^{\frac{1}{r}} }_p \leq \int_{-\infty}^\infty \di t \, \beta_{r}(t)\, \log \norm{ \prod_{k=1}^n A_k^{1 + \ci  t} }_p \quad \textnormal{if} \quad r \in (0,1] \, , 
\end{align}
where $\beta_{r}$ are a family of probability distributions on $\mathbb{R}$, independent of the other parameters. In this article we use the convention that $0^z=0$ for any $z \in \C$.
We refer to Theorem~\ref{thm_ALT_Hirschman} for a  precise statement and discussion. Our extension of the GT inequality again follows in the limit $r \to 0$ by the Lie-Trotter product formula.
We also provide an extension of the ALT and GT inequality for general square matrices (see Theorem~\ref{thm_GT_general}).

We apply our results to quantum information theory and show how it can be used to prove strong sub-additivity. This yields remainder terms on the monotonicity of relative entropy in terms of recoverability, strengthening the Fawzi-Renner bound~\cite{FR14} and subsequent improvements~\cite{BHOS14,TB15,SFR15,wilde15,STH15,JRSWW15}. 
We find that for any positive semi-definite operator $\sigma$, and any trace-preserving completely positive map $\cN$, there exists a trace-preserving completely positive recovery map $\mathcal{R}_{\sigma,\mathcal{N}}$ that satisfies
\begin{align}
D(\rho\|\sigma) - D\big(\cN(\rho) \| \mathcal{N}(\sigma) \big) \geq D_{\mathbb{M}}\big(\rho \|\mathcal{R}_{\sigma,\mathcal{N}}\circ \mathcal{N}(\rho) \big) 
\end{align}
for any quantum state $\rho$. 
Here the bound is given in terms of the measured relative entropy, $D_{\mathbb{M}}(\cdot\|\cdot)$, as in~\cite{STH15}. The recovery map is the explicit universal (i.e., independent of $\rho$) rotated Petz recovery map introduced in~\cite{JRSWW15}. We thus provide the first explicit lower bound that is tight in the commutative case. A precise statement and further results are presented in Section~\ref{sec_app}.

We believe that the proof techniques used in this article, based on asymptotic spectral pinching and complex interpolation theory, yield a transparent method to derive multivariate trace inequalities which should be applicable beyond the extensions of the GT and ALT inequalities studied here. 

Section~\ref{sec_pinching} introduces the method of asymptotic spectral pinching and explains how it can be used to prove trace inequalities. Section~\ref{sec_interpolation} then explains how trace inequalities can be derived via complex interpolation theory. Readers interested in the proof of~\eqref{eq_mainResGT_intro} and~\eqref{eq_ALT_mainres1} may directly proceed to Section~\ref{sec_interpolation}.



\section{Trace inequalities via asymptotic spectral pinching}\label{sec_pinching}

One contribution of this article is the presentation of a transparent method, based on \emph{asymptotic spectral pinching}, that can be used to prove several trace inequalities. The results in this section hold for $p$-norms with $p > 0$ and are in this sense slightly more general then the results mentioned in the introduction. However, the asymptotic spectral pinching method does not yield an explicit form of the distributions $\beta_0$ and $\beta_r$ in~\eqref{eq_mainResGT_intro} and~\eqref{eq_ALT_mainres1}, respectively.


\subsection{The asymptotic spectral pinching method}

Let `$\geq$' denote the L\"owner partial order on positive semi-definite matrices.
Any positive semi-definite matrix $A$ has a decomposition $A=\sum_{\lambda} \lambda P_{\lambda}$ where $\lambda \in \spec(A) \subset \R$ are unique eigenvalues and $P_{\lambda}$ are mutually orthogonal projectors. 
The \emph{spectral pinching map} with respect to $A$ is
\begin{align}
\cP_{A} \ : \  X \mapsto \sum_{\lambda \in \spec(A)}  P_\lambda \,  X \,  P_{\lambda} \ .
\end{align}
Such maps are trace-preserving, completely positive, unital, self-adjoint, and can be viewed as dephasing operations that remove off-diagonal blocks of a matrix. Moreover, they satisfy the following properties: (i) $\cP_{A}[X]$ commutes with $A$ for any $X \geq 0$, (ii) $\tr \, \cP_{A}[X] A =\tr \, XA $ for any $X \geq 0$, and (iii) we have the pinching inequality~\cite{hayashi02},
\begin{align} \label{eq_hayashiPinchingIneqIntro}
\cP_{A}[X] = \sum_{\lambda \in \spec(A)} P_{\lambda}X P_{\lambda}
=\frac{1}{|\spec(A)|} \sum_{y=1}^{|\spec(A)|} U_y X U_y^{\dagger}
 \geq \frac{1}{|\spec(A)|} X \quad \textnormal{for all } X \geq 0  \, ,
\end{align}
where $\spec(A)=\{\lambda_1,\dots,\lambda_{|\spec(A)|}\}$ and $U_y=\sum_{z=1}^{|\spec(A)|} \exp(\frac{\ci 2\pi y z}{|\spec(A)|}) P_{\lambda_z}$ are unitaries. The inequality step in~\eqref{eq_hayashiPinchingIneqIntro} follows form the facts that $U_y X U_y \geq 0$ and $U_{|\spec(A)|} = \id$.
The following observation is crucial. Let $A$ be a positive semi-definite $d \times d$ matrix. The cardinality $|\spec(A^{\otimes m})|$ grows polynomially in $m$ due to the fact that the number of distinct eigenvalues of $A^{\otimes m}$ is bounded by the number of different types of sequences of $d$ symbols of length $m$, a concept widely used in information theory. More precisely~\cite[Lemma~II.1]{csiszar98} gives
\begin{align} \label{eq_types}
|\spec(A^{\otimes m})| \leq 
\left(\begin{matrix} m + d - 1 \\ d - 1 \end{matrix}
\right) 
\leq \frac{(m+d-1)^{d-1}}{(d-1)!}
= O\bigl(\poly(m)\bigr) \ ,
\end{align} 
where $\poly(m)$ denotes an arbitrary polynomial in $m$.
Another useful property of the pinching operation is that it exhibits the following integral representation. 
\begin{lemma} \label{lem_pinchingIntegral}
Let $A$ be positive definite. There exists a probability measure $\mu$ on $\R$ such that
\begin{align} 
\cP_{A}[X]  = \int \mu(\di t) \, A^{\ci t} X A^{-\ci t} \quad \textnormal{for all } X \geq 0  \ .
\end{align}
\end{lemma}
The proof of Lemma~\ref{lem_pinchingIntegral} is given in Appendix~\ref{app_pinchingLemma}.
More information about pinching maps can be found in~\cite[Section~4.4]{carlen_book} or~\cite[Section~2.6.3]{Marco_book}.


\subsection{Illustrative example: intuitive proof of Golden-Thompson inequality}

Here we illustrate the technique by proving the original GT inequality~\eqref{eq_GToriginal}. 
Let $A$ and $B$ be two positive definite matrices. For any $m \in \N$ we find the following chain of inequalities
\begin{align}
\log \tr \exp(\log A + \log B)
&= \frac{1}{m} \log \tr \exp\bigl( \log A^{\otimes m} + \log B^{\otimes m} \bigr) \label{eq_toypinching0} \\
&\leq \frac{1}{m} \log \tr \exp\left( \log \cP_{B^{\otimes m}}[A^{\otimes m}] + \log B^{\otimes m} \right)  + \frac{\log \poly(m)}{m} \label{eq_toypinching1} \\
&=\frac{1}{m}  \log \tr \, \cP_{B^{\otimes m}}[A^{\otimes m}]B^{\otimes m}  + \frac{\log \poly(m)}{m} \label{eq_toypinching2}  \\
&= \log \tr \, A B + \frac{\log \poly(m)}{m} \label{eq_toypinching3}
\ .
\end{align}
The first equality~\eqref{eq_toypinching0} follows because the trace is multiplicative under tensor products.
The sole inequality in~\eqref{eq_toypinching1} follows by the pinching inequality~\eqref{eq_hayashiPinchingIneqIntro}, i.e., Property (iii), together with the fact that the logarithm is operator monotone and $\tr \exp(\cdot)$ is monotone.
Equality~\eqref{eq_toypinching2} uses Property~(i) which ensures that $\cP_{B^{\otimes m}}[A^{\otimes m}]$ commutes with $B^{\otimes m}$ and GT thus holds as an equality for these matrices. Equality~\eqref{eq_toypinching3} employs Property~(ii) and again the multiplicativity of the trace under tensor products.
Considering the limit $m\to \infty$ directly implies the GT inequality~\eqref{eq_GToriginal}.

As we will see later, this proof already suggests an extension of the GT inequality to $n$ matrices by iterative pinching. Let us emphasize the high-level intuition of the proof method presented above. We know that the GT inequality is trivial if the operators commute, and spectral pinching forces our operators to commute. At the same time the pinching should hopefully not destroy the operator which it acts on too much. This is indeed the case (guaranteed by the pinching inequality) if we consider an $m$-fold tensor product of our operators and the limit $m \to \infty$.


\subsection{A convexity result for Schatten quasi-norms}

Let us define the \emph{Schatten} $p$\emph{-norm} of any matrix $L$ as 
\begin{align}
\norm{L}_p:=\big(\tr |L|^p\big)^{\frac{1}{p}} \quad \textnormal{for} \quad p\geq 1 \ ,
\end{align}
where $|L|:=\sqrt{L^\dagger L}$. We extend this definition to all $p>0$, but note that $\norm{L}_p$ is not a norm for $p \in (0,1)$ since it does not satisfy the triangle inequality.
In the limit $p\to \infty$ we recover the \emph{operator norm} and for $p=1$ we obtain the \emph{trace norm}. 

Schatten norms are functions of the singular values and thus unitarily invariant. They satisfy $\|L\|_p = \|L^\dagger\|_p$ and $\|L\|_{2p}^2 = \|LL^\dagger\|_p= \|L^\dagger L\|_p$. They are also multiplicative under tensor products. We note that the Schatten $p$-norm with $p\geq 1$ is the unique norm that is unitarily invariant and multiplicative under tensor products~\cite[Theorem~4.2]{nechita11}.\footnote{Two properties that are crucial for the asymptotic spectral pinching method.} 

Due to the triangle inequality $p$-norms for $p \geq 1$ are convex. In particular, for any probability measure $\mu$ on a measurable space $(\cX,\Sigma)$ and a collection $\{L_x\}_{x \in \cX}$ of matrices, we have
\begin{align} \label{eq_convNorm}
  \norm{\int \mu(\di x) \, L_x}_p \leq 
   \int \mu(\di x) \norm{L_x}_p \,.
\end{align}
Quasi-norms with $p \in (0,1)$ are no longer convex.
However, we show that these quasi-norms still satisfy an asymptotic convexity property for tensor products of matrices in the following sense. We believe that this result may be of independent interest. 
\begin{lemma} \label{lem_normConvexNEW}
Let $p \in (0,1)$, $\mu$ be a probability measure on $(\cX,\Sigma)$ and consider a collection $\{A_x\}_{x \in \mathcal{X}}$ of positive semi-definite matrices. Then
\begin{align} 
\frac{1}{m} \log \norm{\int \mu(\di x) \, A_x^{\otimes m}}_p \leq \frac{1}{m} \log \int \mu(\di x) \norm{A_x^{\otimes m}}_p + \frac{\log \poly (m)}{m}   \, .
\end{align}
\end{lemma}
The proof is given in Appendix~\ref{app_TensorConc}.
Combining this with~\eqref{eq_convNorm} shows that for all $p > 0$ we have the following quasi-convexity property
\begin{align}
\frac{1}{m} \log \norm{\int \mu(\di x) \, A_x^{\otimes m}}_p \leq \log \sup_{x \in \cX} \norm{A_x}_p + \frac{\log \poly (m)}{m} \,.  \label{eq_TCpropnew}
\end{align}


\subsection{Main results and proofs via pinching}

In this section we present two results obtained via the spectral pinching method, which are extensions of the ALT and the GT inequality, respectively, for arbitrarily many matrices. We want to emphasize that in addition to the fact that Theorem~\ref{thm_ALT_pinching} is valid for Schatten quasi-norms, i.e., $p\in (0,1)$, the proof technique via pinching has the advantage of being transparent and intuitive.

\begin{theorem} \label{thm_ALT_pinching}
Let $p> 0$, $r \in (0,1]$, $n\in \N$ and consider a collection $\{A_k\}_{k=1}^n$ of positive semi-definite matrices. Then  
\begin{align}
\norm{ \left| \prod_{k=1}^n A_k^r \right|^{\frac{1}{r}} }_p \leq \sup_{t \in \R^{n}} \norm{ \prod_{k=1}^n A_k^{1+\ci t_k} }_p \,.
\end{align}
\end{theorem}

Before we present the proof, let us given an equivalent statement that follows by a simple substitution $p \leftarrow 2q$ and $A_k \leftarrow \sqrt{A_k}$ for $q > 0$, namely
\begin{multline} \label{eq_eqStat}
  \tr \left( A_1^\frac{r}{2} A_2^\frac{r}{2}  \cdots A_{n-1}^\frac{r}{2}  A_n^r A_{n-1}^\frac{r}{2}  \cdots A_2^\frac{r}{2}  A_1^\frac{r}{2}  \right)^\frac{q}{r} \\
  \leq \sup_{t \in \R^{n-2}} \tr \left( A_1^\frac{1}{2}  A_2^\frac{1+\ci t_2}{2} \cdots A_{n-1}^\frac{1 + \ci t_{n-1}}{2} A_n A_{n-1}^\frac{1 - \ci t_{n-1}}{2} \cdots A_2^\frac{1- \ci t_2}{2} A_1^\frac{1}{2} \right)^{q} \,.
\end{multline}
For $n = 2$ the right-hand side of~\eqref{eq_eqStat} is independent of $t$ and we recover the ALT inequality in~\eqref{eq_ALT_intro}.

\begin{proof}[Proof of Theorem~\ref{thm_ALT_pinching}]
We prove the result for positive definite matrices and note that the generalization to positive semi-definite matrices follows by continuity under the convention that $0^z=0$ for any $z \in \mathbb{C}$.
For convenience of exposition we provide only the proof of Theorem~\ref{thm_ALT_pinching} for three matrices (i.e, $n=3$). The generalization to $n$ matrices follows by appropriately iterating the technical steps presented below.
  Using the multiplicativity of the trace under tensor products, we write
  \begin{align}
  \log \tr \left( A_1^\frac{r}{2} A_2^\frac{r}{2} A_3^r  A_2^\frac{r}{2}  A_1^\frac{r}{2}  \right)^\frac{q}{r} 
   &= \frac{1}{m} \log \tr \left( (A_1^\frac{r}{2})^{\otimes m} (A_2^\frac{r}{2})^{\otimes m} (A_3^r)^{\otimes m}  (A_2^\frac{r}{2})^{\otimes m}  (A_1^\frac{r}{2})^{\otimes m}  \right)^\frac{q}{r} 
  \end{align}
  Then, employing the pinching inequality and the monotonicity of $\tr(\cdot)^{q/r}$, we find
  \begin{align}
    \log \tr \left( A_1^\frac{r}{2} A_2^\frac{r}{2} A_3^r  A_2^\frac{r}{2}  A_1^\frac{r}{2}  \right)^\frac{q}{r} 
    &\!\leq \frac{1}{m} \log \tr \left( (A_1^\frac{r}{2})^{\otimes m} (A_2^\frac{r}{2})^{\otimes m} \cP_{A_2^{\otimes m}}\big[ (A_3^r)^{\otimes m} \big]  (A_2^\frac{r}{2})^{\otimes m}  (A_1^\frac{r}{2})^{\otimes m}  \right)^\frac{q}{r} \!\!   + o(1) \\
    &\!\leq \frac{1}{m} \log \tr \left( (A_1^\frac{r}{2})^{\otimes m} (A_2^\frac{r}{2})^{\otimes m} \cP_{A_2^{\otimes m}}\big[ A_3^{\otimes m} \big]^r  (A_2^\frac{r}{2})^{\otimes m}  (A_1^\frac{r}{2})^{\otimes m}  \right)^\frac{q}{r} \!\!   + o(1) \\
    &\!= \!\frac{1}{m} \log \tr \!\left( (A_1^\frac{r}{2})^{\otimes m} \big((A_2^\frac{1}{2})^{\otimes m} \cP_{A_2^{\otimes m}}\big[ A_3^{\otimes m} \big]  (A_2^\frac{1}{2})^{\otimes m}\big)^r  (A_1^\frac{r}{2})^{\otimes m}  \right)^\frac{q}{r} \!\!\!   + o(1) \ ,
  \end{align}
  where $o(1)$ simply denotes an additive term that vanishes as $m \to \infty$. The second inequality uses the fact that $t \mapsto t^r$ is operator concave for $r \in (0,1]$. The final step uses property (i) of pinching maps. Repeating these steps shows that 
\begin{multline}
 \log \tr \left( A_1^\frac{r}{2} A_2^\frac{r}{2} A_3^r  A_2^\frac{r}{2}  A_1^\frac{r}{2}  \right)^\frac{q}{r}  \\
 \leq  \frac{1}{m} \log \tr \!\left( (A_1^\frac{1}{2})^{\otimes m} \cP_{A_1^{\otimes m}} \big[ (A_2^\frac{1}{2})^{\otimes m} \cP_{A_2^{\otimes m}}\big[ A_3^{\otimes m} \big]  (A_2^\frac{1}{2})^{\otimes m}\big]  (A_1^\frac{1}{2})^{\otimes m}  \right)^q    + o(1) \ .
\end{multline}
The integral representation of pinching maps (see Lemma~\ref{lem_pinchingIntegral}) ensures that there exist probability measures $\mu$ and $\nu$ on $\R$ such that
\begin{align}
 &\log \tr \left( A_1^\frac{r}{2} A_2^\frac{r}{2} A_3^r  A_2^\frac{r}{2}  A_1^\frac{r}{2}  \right)^\frac{q}{r} \nonumber \\
 &\hspace{4mm}\leq\frac{1}{m} \log \tr \left( \int \mu(\di t_1) \int \nu(\di t_2) (A_1^{\frac{1}{2}+\ci t_1})^{\otimes m} (A_2^{\frac{1}{2}+\ci t_2})^{\otimes m} A_3^{\otimes m}  (A_2^{\frac{1}{2}-\ci t_2})^{\otimes m} (A_1^{\frac{1}{2}-\ci t_1})^{\otimes m} \right)^q +o(1) \\
 &\hspace{4mm}\leq \sup_{t \in \R^2} \log \tr \left( A_1^{\frac{1}{2} + \ci t_1} A_2^{\frac{1}{2} + \ci t_2} A_3 A_2^{\frac{1}{2} - \ci t_2} A_1^{\frac{1}{2} - \ci t_1} \right)^q +o(1) \\
 &\hspace{4mm}= \sup_{t \in \R} \log \tr \left( A_1^{\frac{1}{2} } A_2^{\frac{1}{2} + \ci t} A_3 A_2^{\frac{1}{2} - \ci t} A_1^{\frac{1}{2}} \right)^q +o(1) \ ,
\end{align}
where the second inequality uses Lemma~\ref{lem_normConvexNEW}. The final step follows from the fact that Schatten (quasi) norms are unitarily invariant. Considering the limit $m\to \infty$ implies~\eqref{eq_eqStat} and thus completes the proof.
\end{proof}

The multivariate Lie-Trotter product formula (see, e.g.,~\cite[Problem~IX.8.5]{bhatia_book}) states that 
\begin{align} \label{eq_LieTrotter}
  \lim_{r \searrow 0} \big( \exp(r L_1) \exp(r L_2)\, \cdots\, \exp(r L_n) \big)^{\frac{1}{r}} = \exp\left( \sum_{k=1}^n 
  L_k \right)
\end{align}
for square matrices $\{L_k \}_{k=1}^n$. This allows us to derive a multivariate extension of the GT inequality as a limit of the above extended ALT inequality in the limit $r \to 0$.
In particular, combining the product formula with Theorem~\ref{thm_ALT_pinching} implies an extension of the GT inequality to $n$ matrices.

\begin{corollary} \label{cor_GT_pinching}
Let $p> 0$, $n\in \N$ and consider a collection $\{H_k\}_{k=1}^n$ of Hermitian matrices. Then  
\begin{align} \label{eq_ALT_normal}
\norm{ \exp\left( \sum_{k=1}^n H_k \right) }_p \leq \sup_{t \in \R^{n}} \norm{ \prod_{k=1}^n \exp\big( (1+\ci t_k) H_k \big) }_p \,.
\end{align}
\end{corollary}
For $n=2$ the right-hand side term of~\eqref{eq_ALT_normal} is independent of $t$ and we recover the GT inequality~\eqref{eq_GToriginal} for the choice $p=2$. 


\section{Trace inequalities via interpolation theory} \label{sec_interpolation}

For $p$-norms we can prove a more explicit and also more general version of Theorem~\ref{thm_ALT_pinching} based on an entirely different technique\,---\,\emph{complex interpolation theory}.


\subsection{The complex interpolation method} 
The main ingredient for most of our proofs in this section is a complex interpolation result for Schatten norms, commonly attributed to Stein~\cite{stein56}, and based on Hirschman's improvement of the Hadamard three-lines theorem~\cite{H52}. Epstein~\cite{epstein73} showed how interpolation theory can be used in matrix analysis. 

Complex interpolation theory has recently garnered attention in quantum information theory for proving entropy inequalities. Beigi~\cite{beigi13} and Dupuis~\cite{dupuis15} used variations of the Riesz-Thorin theorem based on Hadamard's three line theorem to show properties of sandwiched R\'enyi divergence and conditional R\'enyi entropy, respectively. Wilde~\cite{wilde15} first used these techniques to derive remainder terms for the monotonicity of quantum relative entropy (see Section~\ref{sec_app} for more details). Extensions and further applications of this approach are discussed by Dupuis and Wilde~\cite{DW15}. Hirschmann's refinement was first studied in this context by Junge \emph{et al.}~\cite{JRSWW15}, where the following theorem essentially appeared: 

\begin{theorem}[Stein-Hirschman]\label{thm_hirschman}
Let $S:=\left\{  z\in\mathbb{C}:0\leq\operatorname{Re}(z)
\leq1\right\} $ and let $G$ be a map from $S$ to bounded linear operators on a separable Hilbert space that is holomorphic in the interior of $S$ and continuous on the boundary. Let $p_{0},p_{1}\in [1,\infty]$, $\theta\in(0,1)$, define $p_{\theta}$ by
\begin{equation} 
\frac{1}{p_{\theta}}=\frac{1-\theta}{p_{0}}+\frac{\theta}{p_{1}} \qquad \textrm{and} \qquad \beta_{\theta}(t) := \frac{\sin(\pi\theta)}{2\theta\bigl( \cosh(\pi
t)+\cos(\pi\theta)\bigr)} \, . \label{eq_densi}
\end{equation}
Then, if furthermore $z \mapsto \norm{G(z)}_{p_{\operatorname{Re}(z)}}$ is uniformly bounded on $S$,\footnote{In fact, we only need that $\sup_{z \in S} \exp(-a |\!\operatorname{Im}(z)|) \log \norm{G(z)}_{p_{\operatorname{Re}(z)}} \leq A$ for some constants $A<\infty$ and $a<\pi$.} the following bound holds:
\begin{equation}
\log \left\Vert G(\theta)\right\Vert _{p_{\theta}}  \leq
\int_{-\infty}^{\infty} \di t\ \Bigl(  \beta_{1-\theta}(t)\log   \left\Vert
G(\ci t)\right\Vert _{p_{0}}^{1-\theta} +\beta_{\theta}(t)\log
\left\Vert G(1+\ci t)\right\Vert _{p_{1}}^{\theta} \Bigr) \, .\label{eq:oper-hirschman}%
\end{equation}
\end{theorem}

For the sake of completeness a proof is given in Appendix~\ref{app_Hirschman}. We note that for any $\theta\in(0,1)$ the function $\beta_{\theta}$ is non-negative and
$\int_{-\infty}^\infty \di t\ \beta_{\theta}(t)=1$ 
so that $\beta_{\theta}$ can be interpreted as probability density function on $\mathbb{R}$. These distributions are depicted in Figure~\ref{fig_betaT}.
Furthermore, the following limits hold:
\begin{equation}\label{eq:dist-limit}
\lim_{\theta \searrow 0}\beta_{\theta}(t)=\frac{\pi}{2}\bigl(  \cosh(\pi
t)+1\bigr)^{-1} =:\beta_{0}(t)\ , \qquad \textrm{and} \qquad
\lim_{\theta\nearrow 1}\beta_{\theta}(t)=\delta(t) =:\beta_{1}(t)\,.
\end{equation}
Here $\beta_{0}$ is another probability density function on $\R$ and $\delta(t)$ denotes the Dirac $\delta$-distribution. 

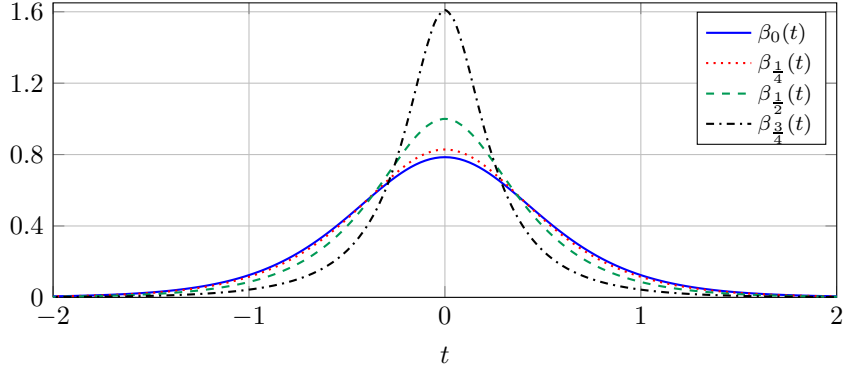
\begin{figure}[!htb]
\centering
  \begin{tikzpicture}
	\begin{axis}[
		height=5.5cm,
		width=12cm,
		grid=major,
		xlabel=$t$,
		xmin=-2,
		xmax=2,
		ymax=1.65,
		ymin=0,
	     xtick={-2,-1,0,1,2},
          ytick={1.6,1.2,0.8,0.4,0},
		legend style={at={(0.904,0.97)},anchor=north,legend cell align=left,font=\footnotesize} 
	]



	\addplot[blue,thick,smooth] coordinates {
(-3.,0.000253484) (-2.95,0.000296591) (-2.9,0.000347026) (-2.85,0.000406036) (-2.8,0.000475078) (-2.75,0.000555855) (-2.7,0.00065036) (-2.65,0.000760925) (-2.6,0.000890277) (-2.55,0.0010416) (-2.5,0.00121863) (-2.45,0.00142572) (-2.4,0.00166796) (-2.35,0.00195131) (-2.3,0.00228272) (-2.25,0.00267032) (-2.2,0.00312361) (-2.15,0.00365367) (-2.1,0.00427342) (-2.05,0.00499796) (-2.,0.00584489) (-1.95,0.00683471) (-1.9,0.00799129) (-1.85,0.00934241) (-1.8,0.0109204) (-1.75,0.0127627) (-1.7,0.0149127) (-1.65,0.0174209) (-1.6,0.0203454) (-1.55,0.0237533) (-1.5,0.0277215) (-1.45,0.0323386) (-1.4,0.0377056) (-1.35,0.0439371) (-1.3,0.051163) (-1.25,0.0595295) (-1.2,0.0691992) (-1.15,0.0803521) (-1.1,0.0931845) (-1.05,0.107908) (-1.,0.124746) (-0.95,0.143929) (-0.9,0.165685) (-0.85,0.190235) (-0.8,0.217769) (-0.75,0.248436) (-0.7,0.282316) (-0.65,0.319396) (-0.6,0.359536) (-0.55,0.402441) (-0.5,0.447625) (-0.45,0.494393) (-0.4,0.541823) (-0.35,0.58877) (-0.3,0.633898) (-0.25,0.675725) (-0.2,0.712712) (-0.15,0.74336) (-0.1,0.766334) (-0.05,0.780573) (0.,0.785398) (0.05,0.780573) (0.1,0.766334) (0.15,0.74336) (0.2,0.712712) (0.25,0.675725) (0.3,0.633898) (0.35,0.58877) (0.4,0.541823) (0.45,0.494393) (0.5,0.447625) (0.55,0.402441) (0.6,0.359536) (0.65,0.319396) (0.7,0.282316) (0.75,0.248436) (0.8,0.217769) (0.85,0.190235) (0.9,0.165685) (0.95,0.143929) (1.,0.124746) (1.05,0.107908) (1.1,0.0931845) (1.15,0.0803521) (1.2,0.0691992) (1.25,0.0595295) (1.3,0.051163) (1.35,0.0439371) (1.4,0.0377056) (1.45,0.0323386) (1.5,0.0277215) (1.55,0.0237533) (1.6,0.0203454) (1.65,0.0174209) (1.7,0.0149127) (1.75,0.0127627) (1.8,0.0109204) (1.85,0.00934241) (1.9,0.00799129) (1.95,0.00683471) (2.,0.00584489) (2.05,0.00499796) (2.1,0.00427342) (2.15,0.00365367) (2.2,0.00312361) (2.25,0.00267032) (2.3,0.00228272) (2.35,0.00195131) (2.4,0.00166796) (2.45,0.00142572) (2.5,0.00121863) (2.55,0.0010416) (2.6,0.000890277) (2.65,0.000760925) (2.7,0.00065036) (2.75,0.000555855) (2.8,0.000475078) (2.85,0.000406036) (2.9,0.000347026) (2.95,0.000296591) (3.,0.000253484)
	};
	\addlegendentry{$\beta_0(t)$}
	
	\addplot[red,thick,smooth,dotted] coordinates {
(-3.,0.000228227) (-2.95,0.00026704) (-2.9,0.000312454) (-2.85,0.000365589) (-2.8,0.000427758) (-2.75,0.000500497) (-2.7,0.000585601) (-2.65,0.000685171) (-2.6,0.000801664) (-2.55,0.000937954) (-2.5,0.0010974) (-2.45,0.00128394) (-2.4,0.00150216) (-2.35,0.00175743) (-2.3,0.00205604) (-2.25,0.00240533) (-2.2,0.00281388) (-2.15,0.0032917) (-2.1,0.0038505) (-2.05,0.00450395) (-2.,0.00526799) (-1.95,0.00616125) (-1.9,0.00720543) (-1.85,0.0084258) (-1.8,0.00985185) (-1.75,0.0115178) (-1.7,0.0134636) (-1.65,0.0157355) (-1.6,0.0183871) (-1.55,0.0214806) (-1.5,0.0250878) (-1.45,0.0292916) (-1.4,0.0341873) (-1.35,0.039884) (-1.3,0.0465066) (-1.25,0.054197) (-1.2,0.0631155) (-1.15,0.0734426) (-1.1,0.0853791) (-1.05,0.0991463) (-1.,0.114985) (-0.95,0.133155) (-0.9,0.153925) (-0.85,0.17757) (-0.8,0.204359) (-0.75,0.234535) (-0.7,0.268297) (-0.65,0.305767) (-0.6,0.346956) (-0.55,0.391718) (-0.5,0.439704) (-0.45,0.490309) (-0.4,0.542633) (-0.35,0.59545) (-0.3,0.647207) (-0.25,0.696069) (-0.2,0.740008) (-0.15,0.776951) (-0.1,0.804966) (-0.05,0.822471) (0.,0.828427) (0.05,0.822471) (0.1,0.804966) (0.15,0.776951) (0.2,0.740008) (0.25,0.696069) (0.3,0.647207) (0.35,0.59545) (0.4,0.542633) (0.45,0.490309) (0.5,0.439704) (0.55,0.391718) (0.6,0.346956) (0.65,0.305767) (0.7,0.268297) (0.75,0.234535) (0.8,0.204359) (0.85,0.17757) (0.9,0.153925) (0.95,0.133155) (1.,0.114985) (1.05,0.0991463) (1.1,0.0853791) (1.15,0.0734426) (1.2,0.0631155) (1.25,0.054197) (1.3,0.0465066) (1.35,0.039884) (1.4,0.0341873) (1.45,0.0292916) (1.5,0.0250878) (1.55,0.0214806) (1.6,0.0183871) (1.65,0.0157355) (1.7,0.0134636) (1.75,0.0115178) (1.8,0.00985185) (1.85,0.0084258) (1.9,0.00720543) (1.95,0.00616125) (2.,0.00526799) (2.05,0.00450395) (2.1,0.0038505) (2.15,0.0032917) (2.2,0.00281388) (2.25,0.00240533) (2.3,0.00205604) (2.35,0.00175743) (2.4,0.00150216) (2.45,0.00128394) (2.5,0.0010974) (2.55,0.000937954) (2.6,0.000801664) (2.65,0.000685171) (2.7,0.000585601) (2.75,0.000500497) (2.8,0.000427758) (2.85,0.000365589) (2.9,0.000312454) (2.95,0.00026704) (3.,0.000228227)
	};
	\addlegendentry{$\beta_{\frac{1}{4}}(t)$}

	\addplot[ForestGreen,thick,smooth,dashed] coordinates {
(-3.,0.000161399) (-2.95,0.000188851) (-2.9,0.000220973) (-2.85,0.000258558) (-2.8,0.000302535) (-2.75,0.000353993) (-2.7,0.000414204) (-2.65,0.000484655) (-2.6,0.000567089) (-2.55,0.000663545) (-2.5,0.000776406) (-2.45,0.000908464) (-2.4,0.00106298) (-2.35,0.00124379) (-2.3,0.00145534) (-2.25,0.00170288) (-2.2,0.00199251) (-2.15,0.00233142) (-2.1,0.00272797) (-2.05,0.00319196) (-2.,0.00373487) (-1.95,0.00437013) (-1.9,0.00511343) (-1.85,0.00598315) (-1.8,0.00700079) (-1.75,0.00819151) (-1.7,0.00958474) (-1.65,0.0112149) (-1.6,0.0131223) (-1.55,0.015354) (-1.5,0.0179651) (-1.45,0.0210202) (-1.4,0.0245945) (-1.35,0.0287761) (-1.3,0.033668) (-1.25,0.0393905) (-1.2,0.0460837) (-1.15,0.0539115) (-1.1,0.0630643) (-1.05,0.0737637) (-1.,0.0862667) (-0.95,0.10087) (-0.9,0.117916) (-0.85,0.137795) (-0.8,0.160949) (-0.75,0.187873) (-0.7,0.219108) (-0.65,0.255231) (-0.6,0.296828) (-0.55,0.344451) (-0.5,0.398537) (-0.45,0.459301) (-0.4,0.526566) (-0.35,0.599546) (-0.3,0.676591) (-0.25,0.75494) (-0.2,0.830584) (-0.15,0.898389) (-0.1,0.952603) (-0.05,0.987789) (0.,1.) (0.05,0.987789) (0.1,0.952603) (0.15,0.898389) (0.2,0.830584) (0.25,0.75494) (0.3,0.676591) (0.35,0.599546) (0.4,0.526566) (0.45,0.459301) (0.5,0.398537) (0.55,0.344451) (0.6,0.296828) (0.65,0.255231) (0.7,0.219108) (0.75,0.187873) (0.8,0.160949) (0.85,0.137795) (0.9,0.117916) (0.95,0.10087) (1.,0.0862667) (1.05,0.0737637) (1.1,0.0630643) (1.15,0.0539115) (1.2,0.0460837) (1.25,0.0393905) (1.3,0.033668) (1.35,0.0287761) (1.4,0.0245945) (1.45,0.0210202) (1.5,0.0179651) (1.55,0.015354) (1.6,0.0131223) (1.65,0.0112149) (1.7,0.00958474) (1.75,0.00819151) (1.8,0.00700079) (1.85,0.00598315) (1.9,0.00511343) (1.95,0.00437013) (2.,0.00373487) (2.05,0.00319196) (2.1,0.00272797) (2.15,0.00233142) (2.2,0.00199251) (2.25,0.00170288) (2.3,0.00145534) (2.35,0.00124379) (2.4,0.00106298) (2.45,0.000908464) (2.5,0.000776406) (2.55,0.000663545) (2.6,0.000567089) (2.65,0.000484655) (2.7,0.000414204) (2.75,0.000353993) (2.8,0.000302535) (2.85,0.000258558) (2.9,0.000220973) (2.95,0.000188851) (3.,0.000161399)
	};
	\addlegendentry{$\beta_{\frac{1}{2}}(t)$}

	\addplot[black,thick,smooth,dashdotted] coordinates {
(-3.,0.0000760929) (-2.95,0.0000890372) (-2.9,0.000104184) (-2.85,0.000121908) (-2.8,0.000142647) (-2.75,0.000166916) (-2.7,0.000195315) (-2.65,0.000228547) (-2.6,0.000267436) (-2.55,0.000312945) (-2.5,0.000366202) (-2.45,0.000428529) (-2.4,0.000501472) (-2.35,0.000586842) (-2.3,0.00068676) (-2.25,0.000803711) (-2.2,0.000940606) (-2.15,0.00110086) (-2.1,0.00128846) (-2.05,0.00150811) (-2.,0.0017653) (-1.95,0.00206648) (-1.9,0.00241924) (-1.85,0.00283247) (-1.8,0.00331662) (-1.75,0.00388401) (-1.7,0.00454912) (-1.65,0.00532901) (-1.6,0.00624384) (-1.55,0.00731738) (-1.5,0.00857781) (-1.45,0.0100585) (-1.4,0.0117991) (-1.35,0.0138469) (-1.3,0.0162583) (-1.25,0.0191009) (-1.2,0.0224558) (-1.15,0.0264213) (-1.1,0.0311164) (-1.05,0.0366861) (-1.,0.0433083) (-0.95,0.0512028) (-0.9,0.0606426) (-0.85,0.0719696) (-0.8,0.0856159) (-0.75,0.102132) (-0.7,0.122225) (-0.65,0.146813) (-0.6,0.177097) (-0.55,0.214659) (-0.5,0.26159) (-0.45,0.320658) (-0.4,0.395477) (-0.35,0.490627) (-0.3,0.611507) (-0.25,0.763405) (-0.2,0.948757) (-0.15,1.16111) (-0.1,1.37577) (-0.05,1.54429) (0.,1.60948) (0.05,1.54429) (0.1,1.37577) (0.15,1.16111) (0.2,0.948757) (0.25,0.763405) (0.3,0.611507) (0.35,0.490627) (0.4,0.395477) (0.45,0.320658) (0.5,0.26159) (0.55,0.214659) (0.6,0.177097) (0.65,0.146813) (0.7,0.122225) (0.75,0.102132) (0.8,0.0856159) (0.85,0.0719696) (0.9,0.0606426) (0.95,0.0512028) (1.,0.0433083) (1.05,0.0366861) (1.1,0.0311164) (1.15,0.0264213) (1.2,0.0224558) (1.25,0.0191009) (1.3,0.0162583) (1.35,0.0138469) (1.4,0.0117991) (1.45,0.0100585) (1.5,0.00857781) (1.55,0.00731738) (1.6,0.00624384) (1.65,0.00532901) (1.7,0.00454912) (1.75,0.00388401) (1.8,0.00331662) (1.85,0.00283247) (1.9,0.00241924) (1.95,0.00206648) (2.,0.0017653) (2.05,0.00150811) (2.1,0.00128846) (2.15,0.00110086) (2.2,0.000940606) (2.25,0.000803711) (2.3,0.00068676) (2.35,0.000586842) (2.4,0.000501472) (2.45,0.000428529) (2.5,0.000366202) (2.55,0.000312945) (2.6,0.000267436) (2.65,0.000228547) (2.7,0.000195315) (2.75,0.000166916) (2.8,0.000142647) (2.85,0.000121908) (2.9,0.000104184) (2.95,0.0000890372) (3.,0.0000760929)
	};
	\addlegendentry{$\beta_{\frac{3}{4}}(t)$}


	\end{axis}  

\end{tikzpicture}
\caption{This plot depicts the probability densities $\beta_\theta$ defined in~\eqref{eq_densi} for $\theta \in \{0,\frac{1}{4},\frac{1}{2},\frac{3}{4}\}$.}
\label{fig_betaT}
\end{figure}



\subsection{Main results and proofs via interpolation theory} \label{sec_main}
In this section we prove our main results which are extensions of the ALT and the GT inequality to arbitrarily many matrices.

\begin{theorem} \label{thm_ALT_Hirschman}
Let $p\geq 1$, $r \in (0, 1]$, $\beta_r$ as defined in~\eqref{eq_densi}, $n\in \N$, and consider a collection $\{A_k\}_{k=1}^n$ of positive semi-definite matrices. Then
\begin{align}\label{eq_ALT_new1}
&\log \norm{ \left| \prod_{k=1}^n A_k^r \right|^{\frac{1}{r}} }_p \leq \int_{-\infty}^{\infty} \di t \, \beta_{r}(t)\, \log \norm{ \prod_{k=1}^n A_k^{1 + \ci  t} }_p \,.
\end{align}
\end{theorem}

\begin{proof}
The case $r = 1$ holds trivially with equality, so suppose $r\in (0,1)$. We prove the result for positive definite matrices and note that the generalization to positive semi-definite matrices follows by continuity. We define the function $G(z):=\prod_{k=1}^n A_k^z = \prod_{k=1}^n \exp(z \log A_k)$ which satisfies the regularity assumptions of Theorem~\ref{thm_hirschman}. Furthermore we pick $\theta = r$, $p_0 = \infty$ and $p_1 = p$ such that $p_{\theta} = \frac{p}{r}$. We find
\begin{align}
\log \norm{G(1+\ci t)}_{p_1}^\theta = r \log \norm{\prod_{k=1}^n A_k^{1+\ci t}}_{p}
\quad \textrm{and} \quad 
\log \norm{G(\ci t)}_{p_0}^{1-\theta} = (1- r) \log \norm{\prod_{k=1}^n A_k^{\ci t} }_{\infty} = 0 \ ,
\end{align}
since the matrices $A_k^{\ci t}$ are unitary.
Moreover, we have
\begin{align}
\log \norm{ G(\theta) }_{p_{\theta}} 
= \log \norm{ \prod_{k=1}^n A_k^{r} }_{\frac{p}{r}} 
= r \log \norm{ \left| \prod_{k=1}^n A_k^{r} \right|^{\frac{1}{r}} }_{p}  \ .
\end{align}
Plugging this into Theorem~\ref{thm_hirschman} yields the desired inequality.
\end{proof}

Let us now remark on several aspects of this inequality. First, we note that the substitution $p \leftarrow 2q$ and $A_k \leftarrow \sqrt{A_k}$ allows to rewrite~\eqref{eq_ALT_new1} in a more suggestive form. For $q \geq \frac12$ and $r \in (0,1]$, we have
\begin{multline} \label{eq_suggestiveForm}
\log \tr \left( A_1^{\frac{r}{2}} A_2^{\frac{r}{2}} \cdots A_{n-1}^{\frac{r}{2}} A_n^{r} A_{n-1}^{\frac{r}{2}} \cdots A_2^{\frac{r}{2}} A_1^{\frac{r}{2}} \right)^{\frac{q}{r}}  \\
 \leq \int_{-\infty}^{\infty} \di t \, \beta_{r}(t)\, \log \tr \left( A_1^{\frac{1}{2}} A_2^{\frac{1+\ci t}{2}} \cdots A_{n-1}^{\frac{1+\ci t}{2}} A_n A_{n-1}^{\frac{1-\ci t}{2}} \cdots A_2^{\frac{1-\ci t}{2}} A_1^{\frac{1}{2}} \right)^q .
\end{multline}
For $n = 2$ the term on the right-hand side is independent of $t$ and we recover the ALT inequality in~\eqref{eq_ALT_intro}. However, we only recover the result for $q \geq \frac12$ using complex interpolation theory. This can be fixed by proving a multivariate extension of the ALT inequality based on pinching (see Theorem~\ref{thm_ALT_pinching}).

Also note that we can always remove the logarithm in the above inequalities by using its concavity and Jensen's inequality. Moreover, for $q \in [\frac{1}{2},1]$ we may pull the integration inside the quasi-norm (by employing the fact that $X \mapsto \log \norm{X}_p$ is concave for $p \in [0,1]$), which yields the following relaxation
\begin{multline}
 \norm{ \left( A_1^{\frac{r}{2}} A_2^{\frac{r}{2}} \cdots A_{n-1}^{\frac{r}{2}} A_n^{r} A_{n-1}^{\frac{r}{2}} \cdots A_2^{\frac{r}{2}} A_1^{\frac{r}{2}} \right)^{\frac{1}{r}} 
}_q \\
 \leq \norm{ \int_{-\infty}^{\infty} \di t \, \beta_{r}(t)\, A_1^{\frac{1}{2}} A_2^{\frac{1+\ci t}{2}} \cdots A_{n-1}^{\frac{1+\ci t}{2}} A_n A_{n-1}^{\frac{1-\ci t}{2}} \cdots A_2^{\frac{1-\ci t}{2}} A_1^{\frac{1}{2}} }_q .
\end{multline}

Next, recall the multivariate Lie-Trotter product formula in~\eqref{eq_LieTrotter}.
Again, this allows us to derive an extension of the GT inequality to arbitrarily many matrices by taking the limit $r \to 0$ of~\eqref{eq_suggestiveForm}.

\begin{corollary} \label{cor_GT_steinHirschman}
Let $p\geq 1$, $\beta_0$ as defined in~\eqref{eq:dist-limit}, $n \in \N$ and consider a collection $\{ H_k\}_{k=1}^n$ of Hermitian matrices. Then
\begin{align} \label{eq_mainResGT}
  \log \norm{ \exp \left( \sum_{k=1}^n H_k \right) }_p \leq 
  \int_{-\infty}^{\infty} \di t \, \beta_0(t)\, \log \norm{ \prod_{k=1}^n  \exp\bigl( (1+\ci t) H_k \bigr) }_p \ \,.
\end{align}
\end{corollary}

Let us take a closer look at the case $n = 3$ and $p = 2$. Substituting $H_k \leftarrow \frac12 H_k$ and using the concavity of the logarithm and Jensen's inequality, we relax Corollary~\ref{cor_GT_steinHirschman} to 
\begin{align} \label{eq_ourGT3}
  \tr \exp \left( H_1 + H_2 + H_3 \right) \leq 
  \int_{-\infty}^{\infty} \di t \,  \beta_0(t)\, \tr\exp(H_1) \exp\left(\tfrac{1+\ci t}{2}H_2\right) \exp(H_3) \exp\left(\tfrac{1-\ci t}{2}H_2\right) .
\end{align} 
This is to be contrasted with Lieb's triple matrix inequality~\cite{Lieb73}, which asserts that
\begin{align} \label{eq_lieb3}
\tr \exp(H_1 + H_2 + H_3)  \leq \int_{0}^{\infty} \!\di \lambda \, \tr  \exp(H_1) \big(\exp(-H_2)+\lambda\,\id \big)^{-1} \exp(H_3) \big(\exp(-H_2)+\lambda\,\id\big)^{-1} \ .
\end{align}
As the next lemma shows, it turns out that these two expressions are in fact equivalent. We believe that this result might be of independent interest as it allows us to write the Fr\'echet derivate of the operator logarithm using an integration over rotations. 

\begin{lemma}\label{lem_LiebRep}
  The following two expressions for the Fr\'echet derivative of the logarithm are equivalent. For any positive definite operator $A$ and Hermitian operator $H$ on a separable Hilbert space, we have
  \begin{align} \label{eq_derivative}
    D\log(A)[H] = \int_{0}^{\infty} \!\di \lambda \,(A+\lambda\,\id)^{-1} H (A+\lambda\,\id)^{-1} = \int_{-\infty}^{\infty} \di t \, \beta_0(t)\, A^{-\frac{1}{2} - \frac{\ci t}{2}} H A^{-\frac{1}{2}+\frac{\ci t}{2}} \,.
  \end{align}
\end{lemma}
The proof is given in Appendix~\ref{app_Derivative}.
The above lemma also gives a further means to understand the probability distribution $\beta_0$ which we obtained from Hirschman's interpolation theorem.\footnote{In Appendix~\ref{app_beta0} we also give numerical evidence that the exact form of the distribution $\beta_0$ is crucial. Our results indicate that if the distribution is more narrow or more flat then Corollary~\ref{cor_GT_steinHirschman} is no longer valid.} 
Whereas Lieb's triple matrix inequality in~\eqref{eq_lieb3} has not been extended to more than three matrices, the alternative representation obtained in~\eqref{eq_ourGT3} trough Corollary~\ref{cor_GT_steinHirschman} naturally extends to arbitrarily many matrices. Finally, it should be noted that Lieb's triple matrix inequality has been shown to be equivalent to many other interesting statements (such as Lieb's concavity theorem~\cite{Lieb73}), and hence it is valuable to have an entirely different proof of these results.


Corollary~\ref{cor_GT_steinHirschman} is valid for Hermitian matrices, but we can extend its scope to general square matrices using the same techniques.

\begin{theorem} \label{thm_GT_general}
Let $p\geq 1$, $\beta_0$ as defined in~\eqref{eq:dist-limit}, $n \in \N$ and consider a collection $\{ L_k\}_{k=1}^n$ of square matrices. Define $\Re(L_k) := \frac12 (L_k + L_k^\dag)$. Then
\begin{align}
  \log \norm{ \exp \left( \sum_{k=1}^n L_k \right) }_p \leq 
  \int_{-\infty}^{\infty} \di t \, \beta_0(t)\, \log \norm{ \prod_{k=1}^n  \exp\bigl( (1+\ci t) \Re(L_k) \bigr) }_p .
\end{align}
\end{theorem}

\begin{proof}
  We write $L_k = \Re(L_k) + \ci \Im(L_k)$ where $\Im(L_k) = \frac1{2\ci}(L_k - L_k^{\dag})$, and note that both $\Re(L_k)$ and $\Im(L_k)$ are Hermitian. Now define 
\begin{align}
  G(z) := \prod_{k=1}^n \exp\big(z \Re(L_k) + \ci \theta \Im(L_k)\big), 
\end{align}
which satisfies the regularity assumption of Theorem~\ref{thm_hirschman}. We note that $G(\ci t)$ is unitary, and thus $\log \norm{G(\ci t)}_{\infty}$ vanishes. We again pick $\theta = r \in (0,1)$, $p_0 = \infty$ and $p_1 = p$ such that $p_{\theta} = \frac{p}{r}$, and find
\begin{align}
  r \log \norm{ \left| \exp \left( r \sum_{k=1}^n L_k \right) \right|^{\frac{1}{r}} }_p &= \log \| G(\theta) \|_{p_{\theta}} \\
  &\leq \int_{-\infty}^{\infty} \di t \, \beta_r(t)\, \log \norm{ G(1 + \ci t) }_p^r \\
  &= r \int_{-\infty}^{\infty} \di t \, \beta_r(t)\, \log \norm{ \prod_{k=1}^n \exp\big((1+\ci t)\Re(L_k) + r \Im(L_k) \big) }_p . \label{eq_ALTgeneral}
\end{align}
Dividing by $r$ and taking the limit $r \to 0$ then yields the desired result via the Lie-Trotter product formula.
\end{proof}

We note that for the case of normal matrices $N$, the matrices $\Re(N)$ and $\Im(N)$ commute, which allows us to slightly simplify the above formula by employing the fact that $\exp(\Re(N)) = \big| \exp(N) \big|$. For two normal matrices the result then reads
\begin{align}
  \norm{ \exp \left( N_1 + N_2 \right) }_p \leq \norm{ \big| \exp(N_1) \big| \big| \exp(N_2) \big| }_p ,
\end{align}
generalizing an inequality by Li and Zhao~\cite{Li14}. Finally, we note that~\eqref{eq_ALTgeneral} can be viewed as an ALT inequality for general square matrices.


\section{An application: entropy inequalities} \label{sec_app}

In this section we show that the multivariate extension of the GT inequality derived in Corollary~\ref{cor_GT_steinHirschman} can be used to derive remainder terms in terms of recoverability for certain entropy inequalities.

For positive semi-definite matrices $\rho,\sigma$ with $\tr\, \rho=1$, Umegaki's \emph{quantum relative entropy}~\cite{umegaki62} is defined as
$D(\rho\|\sigma):= \tr\rho(\log\rho-\log\sigma)$ if $\rho\ll\sigma$ and as $+\infty$ if $\rho \not\ll \sigma$.
Here, $\rho \ll \sigma$ denotes that the support of $\rho$ is contained in the support of $\sigma$. 
We recall the following variational formula for the relative entropy~\cite{Petz_variational88} (see also~\cite{FBT15}):
\begin{align}\label{eq:umegaki_variational}
D(\rho\|\sigma)= \sup_{\omega > 0} \tr\rho \log \omega + 1 - \tr\exp(\log \sigma +  \log \omega) \, .
\end{align}
The \emph{measured relative entropy} is given as
$D_{\mathbb{M}}(\rho\|\sigma):=\sup_{(\cX, M)  } D\big( P_{\rho,M} \big\| P_{\sigma,M} \big)$~\cite{donald86,Pet86,HP91,FBT15},
where the optimization is over positive operator valued measures (POVMs) $M$ on the power-set of a finite set $\cX$, the probability mass functions are given by $P_{\rho,M}(x) = \tr \rho M(x)$, and $D(P\|Q)$ 
is the \emph{Kullback-Leibler divergence}~\cite{kullback51}. 
We recall the following variational formula~\cite{petz_Entropybook,FBT15}: 
\begin{align}\label{eq:measured_variational}
D_{\mathbb{M}}(\rho\|\sigma)=\sup_{\omega > 0} \tr\rho \log \omega+1-\tr\sigma \omega\,.
\end{align}

A fundamental entropy inequality~\cite{LieRus73,lindblad75} states that the quantum relative entropy is monotone under trace-preserving and completely positive maps $\mathcal{N}$, i.e.,\footnote{We note that this monotonicity statement remains valid for more general maps $\cN$~\cite{uhlmann77,hiai10,hayashi_book,reeb16,hermes15}.}   
\begin{align}\label{eq:umegkai_mono}
D(\rho\|\sigma) - D\big(\cN(\rho) \| \mathcal{N}(\sigma) \big) \geq0 \, .
\end{align}
This is closely related to the celebrated strong sub-additivity of quantum entropy~\cite{LieRus73_1,LieRus73} stating that 
\begin{align} \label{eq_SSA}
I(A:C|B)_{\rho} := H(\rho_{AB}) + H(\rho_{BC}) - H(\rho_{B}) - H(\rho_{ABC}) \geq 0
\end{align}
 for any positive semi-definite matrix $\rho_{ABC}$ on a composite Hilbert space $\cH_A \otimes \cH_B \otimes \cH_C$ with ${\tr\, \rho_{ABC}=1}$.
 Here $\rho_{AB}$, $\rho_{BC}$, and $\rho_B$ are marginals of $\rho_{ABC}$ obtained via the partial trace, and $H(\rho) = -\tr \rho \log \rho$ denotes the \emph{von Neumann entropy}.
 
Motivated by recoverability questions in quantum information theory,~\eqref{eq:umegkai_mono} and~\eqref{eq_SSA} have been refined in a series of recent works~\cite{FR14,BHOS14,TB15,SFR15,wilde15,STH15,JRSWW15}, making use of complex interpolation theory as well as asymptotic spectral pinching. With the four matrix extension of the GT inequality given by Corollary~\ref{cor_GT_steinHirschman}, we find the following statement which answers an open question stated in~\cite{JRSWW15}.
\begin{theorem}[Strengthened monotonicity for partial trace]\label{thm:strengthened_mono}
Let $\rho_{AB}$ and $\sigma_{AB}$ be positive semi-definite matrices on $\cH_A \otimes \cH_B$ such that $\rho_{AB} \ll \sigma_{AB}$ and $\tr\, \rho_{AB}=1$. Then
\begin{align}\label{eq:strengthened_mono}
  D(\rho_{AB}\|\sigma_{AB}) - D(\rho_A \| \sigma_A) \geq D_{\mathbb{M}}\big(\rho_{AB} \|\mathcal{R}_{\sigma_{AB},\tr_B}(\rho_A) \big)\, ,
\end{align}
with the rotated Petz recovery map given by
\begin{align}\label{eq:rotated_petz_partial}
\mathcal{R}_{\sigma_{AB},\tr_B}(\cdot):=\int_{-\infty}^{\infty}  \di t \, \beta_0(t)\,\mathcal{R}_{\sigma_{AB},\tr_B}^{\left[t\right]}(\cdot)\quad\mathrm{and}\quad\mathcal{R}_{\sigma_{AB},\tr_B}^{[t]}(\cdot):=\sigma_{AB}^{\frac{1+\ci t}{2}}\left(\sigma_A^{-\frac{1+\ci t}{2}}(\cdot)\sigma_A^{-\frac{1-\ci t}{2}} \otimes \id_B\right)\sigma_{AB}^{\frac{1-\ci t}{2}}.
\end{align}
\end{theorem}
\begin{proof}
Let us recall Corollary~\ref{cor_GT_steinHirschman} applied for $n=4$ and $p=2$. Using the concavity of the logarithm  and Jensen's inequality, it yields 
\begin{align}
  &\tr \exp(H_1+H_2+H_3+H_4) \notag \\
  &\hspace{0mm}\leq \int_{- \infty}^\infty \di t \beta_0(t)\, \tr \exp(H_1) \exp\left(\frac{1\!+\!\ci t}2 H_2 \right) \exp\left(\frac{1\!+\!\ci t}2 H_3 \right) \exp(H_4) \exp\left(\frac{1\!-\!\ci t}2 H_3 \right) \exp\left(\frac{1\!-\!\ci t}2 H_2 \right) \label{eq:GT_4matrix}
\end{align}
for Hermitian matrices $\{H_i\}_{i=1}^4$. Moreover, by definition of the relative entropy for positive definite operators $\rho_{AB}$ and $\sigma_{AB}$, we have
\begin{align}
 D(\rho_{AB}\|\sigma_{AB}) - D(\rho_A \| \sigma_A) = D\bigl( \rho_{AB} \| \exp(\log \sigma_{AB} + \log \rho_A \otimes \id_B - \log \sigma_A \otimes \id_B) \big) \,. \label{eq:contractlog}
\end{align}
For positive semi-definite operators $\rho_{AB}$ and $\sigma_{AB}$, the Hermitian operators $\log \sigma_{AB}$, $\log \rho_A$ and $\log \sigma_A$ are well-defined under the convention $\log 0=0$. Under this convention, the above equality~\eqref{eq:contractlog} also holds for positive semi-definite operators as long as $\rho_{AB} \ll \sigma_{AB}$, which is required by the theorem.  
By the variational formula for the relative entropy~\eqref{eq:umegaki_variational} we thus find 
 \begin{align}
  &D(\rho_{AB}\|\sigma_{AB}) - D(\rho_A \| \sigma_A) \nonumber  \\
 &\qquad= \sup_{\omega_{AB}>0} \tr \, \rho_{AB} \log \omega_{AB} + 1 - \tr \exp(\log \sigma_{AB} + \log \rho_A \otimes \id_B - \log \sigma_A \otimes \id_B + \log \omega_{AB}) \\
 &\qquad\geq\sup_{\omega_{AB}>0} \tr \, \rho_{AB} \log \omega_{AB} + 1 -  \int_{- \infty}^\infty \di t \, \beta_0(t)\, \tr \, \sigma_{AB}^{\frac{1+\ci t}{2}} \left( \sigma_A^{-\frac{1+\ci t}{2}} \rho_A \sigma_A^{-\frac{1-\ci t}{2}}  \otimes \id_B \right) \sigma_{AB}^{\frac{1-\ci t}{2}} \omega_{AB} \\
 &\qquad= D_{\mathbb{M}} \left(\rho_{AB} \middle\| \int_{- \infty}^\infty \di t \beta_0(t) \,  \sigma_{AB}^{\frac{1+\ci t}{2}}  \left( \sigma_A^{-\frac{1+\ci t}{2}} \rho_A \sigma_A^{-\frac{1-\ci t}{2}}  \otimes \id_B \right) \sigma_{AB}^{\frac{1-\ci t}{2}}  \right) \ ,
\end{align}
where the single inequality follows by the four matrix extension of the GT inequality in~\eqref{eq:GT_4matrix}.
The final step uses the variational formula~\eqref{eq:measured_variational} for the measured relative entropy.
\end{proof}

We note that the four matrix extension of the GT inequality is the only inequality used in the proof of Theorem~\ref{thm:strengthened_mono}. 
More properties of the recovery map $\mathcal{R}_{\sigma_{AB},\tr_B}$ given by~\eqref{eq:rotated_petz_partial} are discussed in~\cite{JRSWW15}.

Theorem~\ref{thm:strengthened_mono} implies two other interesting statements. If we substitute $\rho_{AB} \leftarrow \rho_{ABC}$, $\rho_A \leftarrow \rho_{AB}$, $\sigma_{AB} \leftarrow \id_A \otimes \rho_{BC}$, and $\sigma_A \leftarrow \id_A \otimes \rho_B$ we immediately find a remainder term for the conditional quantum mutual information, namely
\begin{align}
I(A:C|B)_{\rho} \geq D_{\mathbb{M}}\big(\rho_{ABC} \big\|\mathcal{R}_{\rho_{BC},\tr_C} \otimes \mathcal{I}_A (\rho_{AB}) \big) ,
\end{align}
where $\mathcal{I}_A$ is the identity map and $\mathcal{R}_{\rho_{BC},\tr_C}$ is defined in~\eqref{eq:rotated_petz_partial}.
Moreover, using the Stinespring dilation theorem~\cite{stinespring55} and the fact that the relative entropy is invariant under isometries, Theorem~\ref{thm:strengthened_mono} generalizes to the following result. 
\begin{corollary}[Strengthened monotonicity]\label{cor:strengthened_mono_d}
Let $\rho,\sigma$ be positive semi-definite matrices such that $\rho \ll \sigma$, $\tr\, \rho=1$, and $\mathcal{N}$ be a trace-preserving completely positive map acting on these matrices. Then
\begin{align}\label{eq:strengthened_mono_d}
  D(\rho\|\sigma) - D\left(\cN(\rho) \middle\| \mathcal{N}(\sigma) \right) \geq D_{\mathbb{M}}\left(\rho \middle\|\mathcal{R}_{\sigma,\mathcal{N}}\circ \mathcal{N}(\rho) \right) ,
\end{align}
with the rotated Petz recovery map given by
\begin{align}\label{eq:rotated_petz}
\mathcal{R}_{\sigma,\mathcal{N}}(\cdot):=\int_{-\infty}^{\infty} \di t \, \beta_0(t)\,\mathcal{R}_{\sigma,\mathcal{N}}^{\left[t\right]}(\cdot)\quad\mathrm{and}\quad\mathcal{R}_{\sigma,\mathcal{N}}^{[t]}(\cdot):=\sigma^{\frac{1+\ci t}{2}}\mathcal{N}^\dagger\left(\mathcal{N}(\sigma)^{-\frac{1+\ci t}{2}}(\cdot)\mathcal{N}(\sigma)^{-\frac{1-\ci t}{2}}\right)\sigma^{\frac{1-\ci t}{2}}.
\end{align}
\end{corollary}
\begin{proof}
  Let us introduce the Stinespring dilation of $\cN$, denoted $U$, and the states $\rho_{AB} = U \rho U^{\dag}$, $\quad \sigma_{AB} = U \sigma U^{\dag}$ such that $\cN(\rho) = \rho_A$ and $\cN(\sigma) = \sigma_B$. Then, using the fact that the relative entropy is invariant under isometries, we have
  \begin{align}
  D(\rho\|\sigma) - D\left(\cN(\rho) \middle\| \mathcal{N}(\sigma) \right) &= D(\rho_{AB}\|\sigma_{AB}) - D(\rho_A\|\sigma_A) \\
  &\geq D_{\mathbb{M}}\big(\rho_{AB} \|\mathcal{R}_{\sigma_{AB},\tr_B}(\rho_A) \big) = D_{\mathbb{M}}\left(\rho \middle\|\mathcal{R}_{\sigma,\mathcal{N}}\circ \mathcal{N}(\rho) \right) ,
  \end{align}
  where the inequality is due to Theorem~\ref{thm:strengthened_mono} and the last equality uses again invariance under isometries and the fact that
  \begin{align}
    U^{\dag} \mathcal{R}_{\sigma_{AB},\tr_B}^{[t]}(\cdot) U &= U^{\dag} U \sigma^{\frac{1+\ci t}{2}} U^{\dag} \left( \mathcal{N}(\sigma)^{-\frac{1+\ci t}{2}} ( \cdot ) \mathcal{N}(\sigma)^{-\frac{1-\ci t}{2}} \otimes \id_B \right) U \sigma^{\frac{1-\ci t}{2}} U^{\dag} U \\
    &= \sigma^{\frac{1+\ci t}{2}} \cN^{\dag} \left( \mathcal{N}(\sigma)^{-\frac{1+\ci t}{2}} ( \cdot ) \mathcal{N}(\sigma)^{-\frac{1-\ci t}{2}} \right) \sigma^{\frac{1-\ci t}{2}}
    = \mathcal{R}_{\sigma,\cN}^{[t]}(\cdot) \,.
  \end{align}
\end{proof}

We note that Corollary~\ref{cor:strengthened_mono_d} is no longer valid if we replace the measured relative entropy in~\eqref{eq:strengthened_mono_d} with a relative entropy. This leads us to believe that~\eqref{eq:strengthened_mono_d} cannot be further improved.

The right-hand side of~\eqref{eq:strengthened_mono_d} can be relaxed using Uhlmann's \emph{fidelity}, $F(\rho,\sigma):=\norm{\sqrt{\rho}\sqrt{\sigma}}_1^2$. It is well known that $D_{\mathbb{M}}(\rho\|\sigma) \geq - \log F(\rho, \sigma)$.\footnote{This follows by the monotonicity of quantum R\'enyi divergence in the order parameter~\cite{MLDSFT13} and of the fact that for any two states there exists an optimal measurement that does not increase their fidelity~\cite[Section~3.3]{Fuc96}.}
Therefore, Corollary~\ref{cor:strengthened_mono_d} implies
\begin{align}
  D(\rho\|\sigma) - D\big(\cN(\rho) \| \mathcal{N}(\sigma) \big) \geq -\log F \big(\rho, \mathcal{R}_{\sigma,\mathcal{N}}\circ \mathcal{N}(\rho) \big) \, .
\end{align}
Moreover, Corollary~\ref{cor:strengthened_mono_d} can be transformed into universal remainder terms (in terms of recoverability with the measured relative entropy) for other entropy inequalities, such as concavity of the conditional entropy and joint convexity of the relative entropy~\cite{BLW14}. We refer to~\cite[Section~5]{JRSWW15} for a more detailed discussion of these bounds.

We also want to refer the reader to Appendix~\ref{app_recovery} where we give a different derivation that yields lower and upper bounds on the difference of relative entropies in Theorem~\ref{thm:strengthened_mono}. This derivation follows the structure of Lieb and Ruskai's original proof of strong sub-additivity~\cite{LieRus73_1,LieRus73}, i.e., it uses the Peierls-Bogoliubov inequality followed by an extension of the GT inequality (see also~\cite{CL14}). However, whereas Lieb and Ruskai use the three matrix extension of the GT inequality (Lieb's triple matrix inequality) we use the four matrix extension of the GT inequality, leading us to a stronger statement.


\section{Discussion} \label{sec_conc}

We discussed two techniques to prove trace inequalities. One is based on asymptotic pinching and the other one uses complex interpolation theory. Both methods lead to transparent and direct proofs of generalized multivariate extensions of the GT and also the more general ALT inequalities. We believe that these methods can be used to prove trace inequalities beyond the extensions of the GT and ALT inequalities studied in this article. For example in~\cite{HP93,ando94}, complementary GT and ALT inequalities have been shown in terms of matrix means. It is left for future research to investigate if these inequalities can be obtained (and possibly be extended to the multivariate case) via pinching or interpolation theory.

Hansen gave an alternative multivariate extension of the GT inequality~\cite{hansen2015} that can be considered an interpolation between the original GT inequality~\eqref{eq_GToriginal} and the operator Jensen inequality~\cite{hansen82,hansen03}. It would be interesting to unify his and our result.
Moreover, Lieb showed that his triple matrix inequality~\eqref{eq_lieb3} is equivalent to many other interesting statements such as several concavity results~\cite{Lieb73}. As Corollary~\ref{cor_GT_steinHirschman} generalizes the triple matrix inequality, it is natural to ask if it can be used to prove more general concavity results.

Ahlswede and Winter noticed that the GT inequality can be used to prove tail bounds for sums of random matrices via the Laplace transform method~\cite{AW02}. As the (original) GT inequality is only valid for two matrices it has to be applied sequentially. Later, Tropp realized that sharper tail bounds can be obtained by using Lieb's concavity theorem instead of the GT inequality~\cite{Tropp11}. An interesting question is whether the multivariate extension of the GT inequality derived in this article (see Corollary~\ref{cor_GT_steinHirschman}) can be used to prove tail bounds for random matrices.


\paragraph*{Acknowledgments.} Elliott H.~Lieb's talk at the ``Beyond I.I.D.~in Information Theory'' workshop in Banff inspired us to derive an alternative proof for his triple matrix inequality. We thank Christian Majenz for allowing us to include his counterexample for a three matrix GT inequality without rotations. 
We would also like to thank J\"urg Fr\"ohlich, Aram W.~Harrow, Spyridon Michalakis, and Renato Renner for useful discussions about trace inequalities. 
We thank Marius C.~Lemm, Thomas Vidick, and Mark M.~Wilde for comments on an earlier draft.
MB and MT thank the Institute for Theoretical Physics at ETH Zurich for hosting them when this project was initiated. 
DS acknowledges support by the Swiss National Science Foundation (SNSF) via the National Centre of Competence in Research ``QSIT'' and by the European Commission via the project ``RAQUEL''.
MB acknowledges funding by the SNSF through a fellowship, funding by the Institute for Quantum Information and Matter (IQIM), an NSF Physics Frontiers Center (NFS Grant PHY-1125565) with support of the Gordon and Betty Moore Foundation (GBMF-12500028), and funding support form the ARO grant for Research on Quantum Algorithms at the IQIM (W911NF-12-1-0521).
MT is funded by an ARC Discovery Early Career Researcher Award (DECRA) fellowship and acknowledges support from the ARC Centre of Excellence for Engineered Quantum Systems (EQUS). 


\appendix


\section{About the probability distribution in Corollary~\ref{cor_GT_steinHirschman} } \label{app_beta0}

As observed in~\cite{thompson65}, we have $\tr \exp(H_1 + H_2 + H_3)  \nleq \tr \exp(H_1) \exp(H_2)\exp(H_3)$ in general. Also a more symmetric conjecture does not hold in general~\cite{majenz16}, i.e.,
\begin{align}
\tr \exp(H_1 + H_2 + H_3)  \nleq \tr \exp\left(H_1\right) \exp\left(\frac12 H_2 \right) \exp\left(H_3\right) \exp\left(\frac12 H_2 \right) .
\end{align}
For the following discussion, let us define 
\begin{align} \label{eq_defGamma}
\gamma(t):=  \tr\,  A_3^{\frac{1}{2}} A_2^{\frac{1}{2} + \frac{\ci t}{2}} A_1 A_2^{\frac{1}{2}-\frac{\ci t}{2}} A_3^{\frac{1}{2}}  \quad \textnormal{and} \quad \kappa:=\tr \exp(\log A_1 + \log A_2 + \log A_3) \ ,
\end{align}
for three positive definite matrices $A_1, A_2$, and $A_3$.
As discussed in~\eqref{eq_ourGT3}, Corollary~\ref{cor_GT_steinHirschman} implies that
\begin{align} \label{eq_shown}
\kappa \leq \int_{-\infty}^\infty \di t \, \beta_0(t) \gamma(t) =:\xi \ .
\end{align}

It is a natural question to investigate how much freedom we have in choosing a probability distribution (different than $\beta_0$) such that~\eqref{eq_shown} remains valid, where the distribution should be independent of the matrices $A_1$, $A_2$, and $A_3$. The following two examples indicate that it might be difficult to find a distribution different than $\beta_0$ that satisfies~\eqref{eq_shown} since it cannot be too narrow (around $t=0$) but also not too flat, either. 
Let us consider the positive semi-definite matrices~\cite{majenz16}:
\begin{align} \label{ex_mat1}
A_1=\frac{1}{4}\left( \begin{matrix}
5 & 2 \\
2  & 1
\end{matrix} \right), \quad 
A_2=\frac{1}{4}\left( \begin{matrix}
1 & -2 \\
-2  & 2
\end{matrix} \right), \quad \textnormal{and} \quad
A_3=\frac{1}{4}\left( \begin{matrix}
8 & -2 \\
-2  & 1
\end{matrix} \right)\ .
\end{align}
As a second example we consider the positive semi-definite matrices
\begin{align} \label{ex_mat2}
A_1=\frac{1}{8}\left( \begin{matrix}
4 & 2-\ci \\
2+\ci  & 3
\end{matrix} \right), \,\, \,\,
A_2=\frac{1}{60}\left( \begin{matrix}
15 & -5-3\ci \\
 -5+3\ci  & 12
\end{matrix} \right), \,\,\,\, \textnormal{and} \,\,\,\,
A_3=\frac{1}{20}\left( \begin{matrix}
15 & 10-5\ci \\
 10+5\ci  & 11
\end{matrix} \right) \ .
\end{align}
Figure~\ref{fig_beta} compares $\kappa$ with $\gamma(t)$. We note that the matrices~\eqref{ex_mat1} also show that $\kappa > \gamma(0)$ is possible, i.e., a three matrix extension of the GT inequality without any phases does not hold in general~\cite{majenz16}.
\begin{figure}[!htb]
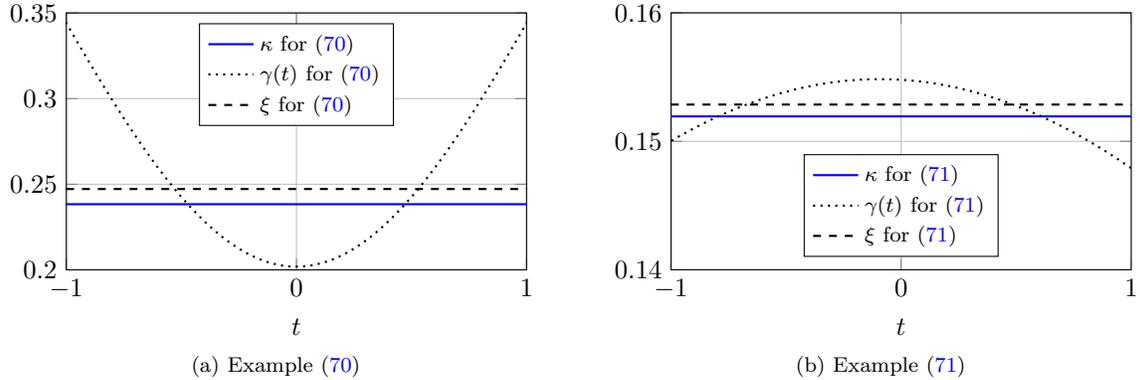

\centering
\subfloat[Example~\eqref{ex_mat1}]{ \input{beta1.tex} } \hspace{5mm}
\subfloat[Example~\eqref{ex_mat2}]{\input{beta2.tex}}\\
\caption{This plot compare $\gamma(t)$ with $\kappa$ for the matrices defined in~\eqref{ex_mat1} and~\eqref{ex_mat2}. If we want $\kappa \leq \int \mu(\di t) \gamma(t)$ to hold for some probability measure $\mu$ on $\R$ that does not depend on $A_1$, $A_2$ and $A_3$, these two example show that $\mu$ cannot be too narrow (around $t=0$) but also not too flat, either.}
\label{fig_beta}
\end{figure}

\section{Proof of Lemma~\ref{lem_pinchingIntegral}} \label{app_pinchingLemma}

We want to write $\cP_{A}[X]= \sum_{\lambda \in \spec(A)} P_{\lambda}X P_{\lambda}$ in the from $\cP_{A}[X] = \int \mu(\di t) A^{\ci t} X A^{-\ci t}$, where $\mu$ is a probability measure on $\R$. Recalling that $A=\sum_{\lambda \in \spec(A)} \lambda P_\lambda$ and the fact that the eigenvectors to distinct eigenvalues of positive semi-definite matrices are orthogonal we find
\begin{align}
A^{\ci t} 
= \sum_{\lambda \in \spec(A)} \lambda^{\ci t} P_\lambda \ .
\end{align}
and thus
\begin{align}
A^{\ci t} X A^{-\ci t} = \sum_{\lambda,\lambda' \in \spec(A)} \lambda^{\ci t} (\lambda')^{-\ci t} P_\lambda X P_{\lambda'} \ .
\end{align}
Defining $\tilde \lambda = \log \lambda$ and $\tilde \lambda' = \log \lambda'$ we thus require 
\begin{align} \label{eq_FT}
\hat \mu(\tilde \lambda' - \tilde \lambda) = \int \mu(\di t) \exp\big(- \ci t(\tilde \lambda' - \tilde \lambda) \big) = \int \mu(\di t) \lambda^{\ci t} (\lambda')^{-\ci t} = \delta\{\lambda=\lambda' \} = \delta\{ \tilde \lambda = \tilde \lambda'Ã\} \ ,
\end{align}
where $\hat \mu$ denotes the Fourier transform of $\mu$. We thus require that (i) $\hat \mu(0)=1$ and (ii) $\hat \mu(\tilde \lambda - \tilde \lambda')=0$ for all $\exp(\tilde \lambda), \exp(\tilde \lambda') \in \spec(A)$. Let us define
\begin{align}
\Delta:= \min\{|\tilde \lambda - \tilde \lambda'|: \exp(\tilde \lambda), \exp(\tilde \lambda') \in \spec(A), \tilde \lambda \ne \tilde \lambda' \} \ .
\end{align}
Furthermore for a fixed $\tau> 0$ we define the triangular function
\begin{align}
r_{\tau}(t)=\left \lbrace \begin{array}{l l}
1-\frac{|t|}{\tau},& \textnormal{if } |t| < \tau \\
0 & \textnormal{otherwise} \ .
\end{array} \right.
\end{align}
We next pick $\hat \mu(\xi) = r_{\frac{\Delta}{2}}(\xi)$ which clearly satisfies the requirements (i) and (ii) mentioned above. Its inverse Fourier transform can be computed as
\begin{align} \label{eq_constructedMU}
\mu(t) = \frac{1}{2\pi} \int_{-\infty}^\infty \di \xi\, r_{\frac{\Delta}{2}}(\xi) \exp(\ci t \xi) = \frac{1-\cos(\Delta\, t / 2 )}{\pi \Delta \, t^2 /2} \ .
\end{align}
It is immediate to verify that $\mu$ as given in~\eqref{eq_constructedMU} is a probability distribution on $\R$ (i.e., $\mu(t) \geq 0$ for all $t \in \R$ and  $\int_{\R} \di t \, \mu(t) =1$). This is the distribution that satisfies the assertion of Lemma~\ref{lem_pinchingIntegral} and thus completes the proof.


\section{Proof of Lemma~\ref{lem_normConvexNEW}} \label{app_TensorConc}

Let $\cH$ denote the Hilbert space of dimension $d$ where the matrices $A_x$ act on.
For any $x$, consider the spectral decomposition $A_x = \sum_{k} \lambda_k \ket{k}\!\bra{k}$ in Dirac bra-ket notation. Introducing an isometric space $\cH'$, we define the vector $\ket{v_x} \in \cH \otimes \cH'$ by $\ket{v_x} = \sum_k \sqrt{\lambda_k} \ket{k} \otimes \ket{k}$\,---\,i.e., the purification of $A_x$. Now note that
the projectors $( \ket{v_x}\bra{v_x} )^{\otimes m}$ lie in the symmetric subspace of $(\cH \otimes \cH')^{\otimes m}$ whose dimension grows as $\poly(m)$.\footnote{This follows from the fact that the dimension of the symmetric subspace of $\cH^{\otimes m}$ is equal to the number of types of sequences of $d$ symbols of length $m$, which is polynomial in $m$ (as shown in~\eqref{eq_types}).} 
Moreover, we have
\begin{align}
\int \mu(\di x) A_k^{\otimes m} =  \int \mu(\di x) \tr_{\cH^{'\otimes m}} \left(  \ket{v_x}\bra{v_x} \right)^{\otimes m} \,.
\end{align}
Then by Carath\'eodory's theorem (see, e.g.,~\cite[Theorem 18]{eggleston_book}) there exists a discrete probability measure~$P$ on $\in \cI \subset \cX$ with $|\cI| = \poly(m)$ such that
\begin{align}
\int \mu(\di x) A_x^{\otimes m} = \sum_{x \in \cI} P(x) A_x^{\otimes m} 
\quad \textrm{and} \quad
\int \mu(\di x) \norm{A_x^{\otimes m}}_p = \sum_{x \in \cI} P(x) \norm{A_x^{\otimes m}}_p \,.\label{eq_sd1}
\end{align}

If $p\in (0,1)$ the Schatten $p$-norms only satisfy a weakened version of the triangle inequality (see, e.g.,~\cite[Equation~20]{Kittaneh97}), which states that
\begin{align}
  \norm{\sum_{x=1}^n A_x }_p^p \leq \sum_{x=1}^n \norm{ A_x }_p^p \,.
\end{align}
Hence, we find the following chain of inequalities
\begin{align}
\frac{1}{m} \log \norm{\int \mu(\di x) A_x^{\otimes m}}_p 
&=\frac{1}{m} \log \norm{\sum_{x \in \cI} P(x)A_x^{\otimes m}}_p \\
&\leq \frac{1}{m} \log \left( \sum_{x \in \cI} \norm{P(x) A_x^{\otimes m}}_p^p \right)^{\frac{1}{p}} \\
&= \frac{1}{m} \log \left( |\cI|^\frac{1}{p} \Big( \frac{1}{|I|} \sum_{x \in I} \norm{P(x) A_x^{\otimes m}}_p^p \Big)^{\frac{1}{p}} \right) \\
&\leq \frac{1}{m} \log\left(|\cI|^{\frac{1}{p}-1} \sum_{x\in \cI} \norm{P(x) A_x^{\otimes m}}_p \right) \\
&= \frac{1}{m} \log \left( \sum_{x \in \cI} P(x) \norm{A_x^{\otimes m}}_p \right) + \frac{1}{m} \frac{1-p}{p} \log|\cI| \\
&= \frac{1}{m} \log \left( \int \mu(\di x) \norm{A_x^{\otimes m}}_p \right) + \frac{\log \poly(m)}{m} \ ,
\end{align}
where the second inequality uses that the map $t \mapsto t^{\frac{1}{p}}$ is convex for $p\in (0,1)$. The final step follows from the fact that $|\cI| = \poly(m)$.

\section{Proof of Theorem~\ref{thm_hirschman}} \label{app_Hirschman}

We follow the argument given in~\cite[Appendix A]{JRSWW15}, and take care of the explicit conditions on the Schatten norms of $G(z)$.
We recall Hirschman's strengthening~\cite{H52} (see also~\cite[Lemma~1.3.8]{G08}) of Hadamard's three line theorem.
\begin{lemma}[Hirschman]
\label{lm:hirschman} Let $S:=\left\{  z\in\mathbb{C}:0\leq\operatorname{Re}(z) \leq1\right\}$ and let $g(z)$ be uniformly bounded on $S$, holomorphic on the interior of $S$ and continuous on the
boundary. Then for $\theta\in(0,1)$, we have
\begin{equation}
\log\left\vert g(\theta)\right\vert \leq\int_{-\infty}^{\infty} \di t\ 
\beta_{1-\theta}(t)\log  \left\vert g(\ci t)\right\vert ^{1-\theta}
+\beta_{\theta}(t)\log  \left\vert g(1+\ci t)\right\vert ^{\theta}
 \,.
\end{equation}
Moreover, the assumption that the function is uniformly bounded can be relaxed to
\begin{align}
\sup_{z \in S} \exp\big(-a |\!\operatorname{Im}(z)| \big) \log |g(z)| \leq A  \quad \text{for some constants} \quad A<\infty \quad \text{and} \quad a<\pi \, .
\end{align}
\end{lemma}

We are now prepared to prove Theorem~\ref{thm_hirschman}.

\begin{proof}[Proof of Theorem~\ref{thm_hirschman}]
 For $x \in [0,1]$, define $q_{x}$ as the H\"older conjugate of $p_{x}$ such that $p_{x}^{-1} + q_{x}^{-1} = 1$. Hence, using the definition of $p_x$ in~\eqref{eq_densi}, we have
\begin{align}
  \frac{1}{q_x} = \frac{1-x}{q_0} + \frac{x}{q_1} \,.
\end{align}

Now for our fixed $\theta \in (0,1)$ the operator $G(\theta)$ is bounded by assumption and thus allows a polar decomposition, $G(\theta) = U \Delta$, where $\Delta$ is positive semi-definite and $U$ is a partial isometry~\cite[Theorem~VI.10]{simon_book} satisfying $\Delta U^{\dag}U = U^{\dag}U \Delta = \Delta$. Then define $X(z)$ via
\begin{align}
  X(z)^{\dag} =  C^{-p_{\theta} \left( \frac{1-{z}}{q_0} + \frac{{z}}{q_1} \right)} \Delta^{ p_{\theta} \left( \frac{1-{z}}{q_0} + \frac{{z}}{q_1} \right)  } U^{\dag}
  \qquad \textrm{with} \qquad C := \norm{ \Delta }_{p_{\theta}} = \norm{ G(\theta) }_{p_{\theta}} < \infty \,.
\end{align}
We find that $z \mapsto X(z)$ is anti-holomorphic on $S$ and
\begin{align}
  \norm{ X(x + \ci y) }_{q_x}^{q_x} = \tr\, \left( C^{-1} \Delta \right)^{p_{\theta}q_{x}\left( \frac{1-x}{q_0} + \frac{x}{q_1} \right)} = \tr\, \left( C^{-1} \Delta \right)^{p_{\theta}} = 1 \,.
\end{align}

Consequently, the Hilbert-Schmidt inner product $g(z) := \tr\, X(z)^{\dag}G(z)$ is holomorphic and bounded on $S$ because the H\"older inequality (see, e.g.,~\cite[Theorem~7.8]{weidmann_book}) yields
\begin{align}
   |g(x + \ci y)| \leq \norm{ X(x + \ci y) }_{q_x}  \norm{ G(x + \ci y) }_{p_x} \leq \norm{ G(x + \ci y) }_{p_x}  ,
   \label{eq:hoelder-used} \,.
\end{align}
Hence, our assumptions on $G(z)$ imply that $g(z)$ satisfies the assumptions of Lemma~\ref{lm:hirschman}.

It remains to verify the following relations using the H\"older inequality in~\eqref{eq:hoelder-used}:
\begin{align}
  g(\theta) &= \tr\, X(\theta)G(\theta) = C^{-p_{\theta}\frac{1}{q_{\theta}}}\, \tr\, \Delta^{p_{\theta}-1} U^{\dag} U \Delta = C^{1-p_{\theta}}\, \tr\, \Delta^{p_{\theta}} = \norm{G(\theta)}_{p_{\theta}} \,, \\
  |g(\ci t)| &\leq \norm {G(\ci t)}_{p_0} \,, \qquad \textrm{and} \qquad |g(1+\ci t)| \leq \norm {G(1+\ci t)}_{p_1}
  \,.
\end{align}
Substituting this into Lemma~\ref{lm:hirschman} yields the desired result.
\end{proof}



\section{Proof of Lemma~\ref{lem_LiebRep}} \label{app_Derivative}

The first expression for the derivative given in~\eqref{eq_derivative} is well known and can be derived using integral representations of the operator logarithm (see, e.g., \cite{carlen_book}). Now let $A=\sum_k\mu_k\ket{k}\!\bra{k}$ for an orthonormal eigenbasis $\{ \ket{k} \}_k$ of $A$. The claim is thus equivalent to
\begin{align}
\int_{-\infty}^{\infty} \di t \, \beta_0(t)\, \mu_i^{-\frac{1}{2} - \frac{\ci t}{2}}\mu_j^{-\frac{1}{2}+\frac{\ci t}{2}}\bra{k}H\ket{\ell}=\int_{0}^{\infty} \!\di \lambda\left(\mu_i+\lambda\right)^{-1}\left(\mu_j+\lambda\right)^{-1}\bra{k}H\ket{\ell}\qquad \forall\, k,\ell \,.
\end{align}
Thus, it suffices to show that
\begin{align}
\frac{1}{\sqrt{xy}}\int_{-\infty}^{\infty} \di t \, \beta_0(t)\,\left(\frac{y}{x}\right)^{\frac{\ci t}{2}}=\int_{0}^{\infty} \!\di \lambda\left(x+\lambda\right)^{-1}\left(y+\lambda\right)^{-1} \qquad \forall\, x, y > 0 \,.
\end{align}
Since $\beta_0(t)$ is symmetric in $t$, we have
\begin{align}
\frac{1}{\sqrt{xy}} \int_{-\infty}^{\infty} \di t \, \beta_0(t)\,\left(\frac{y}{x}\right)^{\frac{\ci t}{2}}
=\frac{1}{\sqrt{xy}} \int_{-\infty}^{\infty} \di t \, \beta_0(t)\,\cos\left(\frac{t}{2}\log\left(\frac{y}{x}\right)\right) 
=\frac{1}{y-x}\log\left(\frac{y}{x}\right) ,
\end{align}
and the claim follows because we also have
\begin{align}
\int_{0}^{\infty} \!\di \lambda\left(x+\lambda\right)^{-1}\left(y+\lambda\right)^{-1}=\frac{1}{y-x}\log\left(\frac{y}{x}\right) .
\end{align}


\section{Additional recoverability bounds } \label{app_recovery}
The purpose of this section is to present two additional entropy inequalities that also follow from the multivariate extension of the GT inequality given by Corollary~\ref{cor_GT_steinHirschman}. These two bounds have been proven before~\cite{JRSWW15,DW15}.

\begin{proposition} \label{prop_addBounds}
Let $\rho_{AB}$ and $\sigma_{AB}$ be positive semi-definite matrices on $\cH_A \otimes \cH_B$ such that $\rho_{AB} \ll \sigma_{AB}$ and $\tr\, \rho_{AB}=1$ and let $\cR_{\sigma_{AB},\tr_B}^{[t]} $ be as defined in~\eqref{eq:rotated_petz_partial}. Then
\begin{align}
- \int_{-\infty}^\infty \di t\, \beta_0(t) \log F\big(\rho_{AB}, \cR_{\sigma_{AB},\tr_B}^{[t]}(\rho_A)\big) 
&\leq D(\rho_{AB} \| \sigma_{AB}) - D(\rho_A \| \sigma_A) \label{eq_MarkLB} \\
&\leq  \int_{-\infty}^\infty \di t\, \beta_0(t) D_2\big(\rho_{AB} \big\| \cR_{\sigma_{AB},\tr_B}^{[t]}(\rho_A)\big) \ , \label{eq_MarkUB}
\end{align}
where $D_2(\rho\|\sigma):=\log \tr \rho^2 \sigma^{-1}$ is Petz' R\'enyi relative entropy of order $2$.
\end{proposition}

\begin{proof}
Klein's inequality~\cite{klein31} states that for any Hermitian matrices $H_1$, $H_2$ and for any differentiable convex function $f:\R \to \R$, we have
$\tr\bigl(f(H_1)-f(H_2) - (H_1-H_2)f'(H_2)\bigr)\geq 0$.
If we apply Klein's inequality with $f(\cdot)=\exp(\cdot)$, $H_1=G_1+G_2$ and $H_2=G_1+ \id \, \tr\, G_2\exp(G_1)$ we obtain the Peierls-Bogoliubov inequality (see, e.g.,~\cite{ruelle_book}) which tells us that for Hermitian matrices $G_1$ and $G_2$, we have
\begin{align} \label{eq_PBineq}
-\tr \, G_2\exp(G_1) \geq- \log \tr \exp(G_1+G_2)\,.
\end{align}

We first prove~\eqref{eq_MarkLB}.
Let $\rho_{AB}$ and $\sigma_{AB}$ be positive semi-definite matrices on $\cH_A \otimes \cH_B$ such that $\rho_{AB} \ll \sigma_{AB}$ and $\tr \rho_{AB} =1$.
For $G_1 =\log\rho_{AB}$ and $G_2 =\frac{1}{2}(\log\rho_{A} \otimes \id_B +\log\sigma_{AB}-\log\sigma_{A} \otimes \id_B -\log\rho_{AB})$, this gives
\begin{align}
&D(\rho_{AB} \| \sigma_{AB}) - D(\rho_A \| \sigma_A) \\
&\qquad\geq-2\log\tr \, \exp\left(\frac{1}{2}\left(\log\rho_{A} \otimes \id_B+\log\sigma_{AB}-\log\sigma_{A} \otimes \id_B +\log\rho_{AB}\right)\right)\\
&\qquad\geq-2 \int_{-\infty}^\infty \!\! \di t\,  \beta_0(t)\, \log \norm{\rho_{AB}^{\frac{1+\ci t}{2}}\sigma_{AB}^{\frac{1+\ci t}{2}} \Big( \sigma_{A}^{-\frac{1+\ci t}{2}}\rho_{A}^{\frac{1+\ci t}2} \otimes \id_B \Big)}_{1}\\
&\qquad=-  \int_{-\infty}^\infty \!\! \di t\, \beta_0(t) \log F\left(\rho_{AB},\sigma_{AB}^{\frac{1+\ci t}{2}} \Big(\sigma_{A}^{-\frac{1+\ci t}{2}}\rho_{A}\sigma_{A}^{-\frac{1-\ci t}{2}} \otimes \id_B \Big) \sigma_{AB}^{\frac{1-\ci t}{2}}\right)\label{eq_stopLB}   ,
\end{align}
where the penultimate step uses the extension of the GT inequality from Corollary~\ref{cor_GT_steinHirschman} for $n=4$ and $p=1$.  

It remains to prove~\eqref{eq_MarkUB}.
Applying the Peierls-Bogoliubov inequality~\eqref{eq_PBineq} for $G_1 = \log \rho_{AB}$ and $G_2 = \log \sigma_A \otimes \id_B + \log \rho_{AB} - \log \rho_{A} \otimes \id_B - \log \sigma_{AB}$, we find
\begin{align}
D(\rho_{AB} \| \sigma_{AB}) - D(\rho_A \| \sigma_A) &\leq  \log \tr \exp(2\log \rho_{AB} +\log \sigma_A - \log \rho_{A} - \log \sigma_{AB}) \\
&\leq \int_{-\infty}^\infty \di t \, \beta_0(t) \, \log \tr \rho_{AB}^2 \sigma_{AB}^{-\frac{1+\ci t}{2}} \Big( \sigma_A^{\frac{1+\ci t}{2}} \rho_{A}^{-1} \sigma_A^{\frac{1-\ci t}{2}} \otimes \id_B \Big) \sigma_{AB}^{-\frac{1-\ci t}{2}} \\
&= \int_{-\infty}^\infty \di t \, \beta_0(t) \, D_2\left(\rho_{AB} \middle\| \sigma_{AB}^{\frac{1+\ci t}{2}} \Big( \sigma_A^{-\frac{1+\ci t}{2}} \rho_{A} \sigma_A^{\frac{1-\ci t}{2}} \otimes \id_B \Big) \sigma_{AB}^{\frac{1-\ci t}{2}}\right)  ,
\end{align}
where the second inequality follows by Corollary~\ref{cor_GT_steinHirschman} applied for $n=4$ and $p=2$.
\end{proof}
Following the same line of arguments as in the proof in the proof of Corollary~\ref{cor:strengthened_mono_d},~\eqref{eq_MarkLB} can be extended to the case of arbitrary trace-preserving completely positive maps. This then reproduces a result in~\cite[Section 3]{JRSWW15}.

\bibliographystyle{arxiv_no_month}
\bibliography{bibliofile}
      
\end{document}